\documentclass[onecolumn, 11pt]{IEEEtran}
\usepackage{amsmath}
\usepackage{amssymb}
\usepackage{amsfonts}
\usepackage{graphicx}
\usepackage{epsfig}
\usepackage{subfigure}
\usepackage{psfrag}

\title{Achieving Global Optimality for Weighted Sum-Rate Maximization in the $K$-User Gaussian Interference Channel with Multiple
Antennas \footnote{L. Liu is with the Department of Electrical and
Computer Engineering, National University of Singapore
(e-mail:liu\_liang@nus.edu.sg).} \footnote{R. Zhang is with the
Department of Electrical and Computer Engineering, National
University of Singapore (e-mail: elezhang@nus.edu.sg). He is also
with the Institute for Infocomm Research, A*STAR, Singapore.}
\footnote{K. C. Chua is with the Department of Electrical and
Computer Engineering, National University of Singapore
(e-mail:eleckc@nus.edu.sg).}}

\author{Liang Liu, Rui Zhang, and Kee-Chaing Chua}

\begin{document}
\maketitle \thispagestyle{empty}

\begin{abstract}
Characterizing the global maximum of weighted sum-rate (WSR) for the
$K$-user Gaussian interference channel (GIC), with the interference
treated as Gaussian noise, is a key problem in wireless
communication. However, due to the users' mutual interference, this
problem is in general non-convex and thus cannot be solved directly
by conventional convex optimization techniques. In this paper, by
jointly utilizing the monotonic optimization and rate profile
techniques, we develop a new framework to obtain the globally
optimal power control and/or beamforming solutions to the WSR
maximization problems for the GICs with single-antenna transmitters
and single-antenna receivers (SISO), single-antenna transmitters and
multi-antenna receivers (SIMO), or multi-antenna transmitters and
single-antenna receivers (MISO). It is assumed that the transmitted
signals have circularly symmetric complex Gaussian distributions and
are independent over time. Different from prior work, this paper
proposes to maximize the WSR in the achievable rate region of the
GIC directly by exploiting the facts that the achievable rate region
is a ``normal'' set and the users' WSR is a ``strictly increasing''
function over the rate region. Consequently, the WSR maximization is
shown to be in the form of monotonic optimization over a normal set
and thus can be solved globally optimally by the existing outer
polyblock approximation algorithm. However, an essential step in the
algorithm hinges on how to efficiently characterize the intersection
point on the Pareto boundary of the achievable rate region with any
prescribed ``rate profile'' vector. This paper shows that such a
problem can be transformed into a sequence of
signal-to-interference-plus-noise ratio (SINR) feasibility problems,
which can be solved efficiently by existing techniques. Numerical
results validate that the proposed algorithms can achieve the global
WSR maximum for the SISO, SIMO or MISO GIC, which serves as a
performance benchmark for existing heuristic algorithms.
\end{abstract}


\setlength{\baselineskip}{1.0\baselineskip}
\newtheorem{definition}{\underline{Definition}}[section]
\newtheorem{fact}{Fact}
\newtheorem{assumption}{Assumption}
\newtheorem{theorem}{\underline{Theorem}}[section]
\newtheorem{lemma}{\underline{Lemma}}[section]
\newtheorem{corollary}{\underline{Corollary}}[section]
\newtheorem{proposition}{\underline{Proposition}}[section]
\newtheorem{example}{\underline{Example}}[section]
\newtheorem{remark}{\underline{Remark}}[section]
\newtheorem{algorithm}{\underline{Algorithm}}[section]
\newcommand{\mv}[1]{\mbox{\boldmath{$ #1 $}}}

\section{Introduction}\label{sec:introduction}

Gaussian interference channel (GIC) is a basic mathematical model
that characterizes many real-life interference-limited communication
systems. The information-theoretic study on the GIC has a long
history, but the capacity region of the GIC still remains unknown in
general, even for the two-user case. The best achievable rate region
for the two-user GIC to date was established by Han and Kobayashi in
\cite{Han81}, which utilizes rate splitting at transmitters, joint
decoding at receivers, and time sharing among codebooks. This
achievable rate region was recently proven to be within 1-bit of the
capacity region of the GIC in \cite{Etkin08}. However,
capacity-approaching techniques in general require non-linear
multi-user encoding and decoding, which may not be suitable for
practical systems. A more pragmatic approach that leads to
suboptimal achievable rates is to allow only single-user encoding
and decoding by treating the interference from all other unintended
users as additive Gaussian noise. For this approach, the key design
challenge lies in how to optimally allocate transmit resources such
as power, bandwidth, and antenna beam among different users to
minimize the network performance loss due to their mutual
interference. Recently, \cite{Jafar09} showed that the circularly
symmetric complex Gaussian (CSCG) distribution for the transmitted
signals is in general non-optimal for the rate maximization in GIC
with the interference treated as noise. By means of symbol
extensions over time and/or asymmetric complex signaling, the
weighted sum-rate (WSR) of GIC can be further improved. However, to
our best knowledge, applying such techniques will result in more
complicated WSR maximization problems, for which how to obtain the
globally optimal solutions still remains an open problem, even for
the case of $2$-user GIC. Thus, for simplicity, in this paper we
adopt the conventional assumption for the GIC that the transmitted
signals have an independent CSCG distribution over time.

The research on the GIC with interference treated as noise has
recently drawn significant attention due to the advance in
cooperative inter-cell interference (ICI) management for cellular
networks. Traditionally, most of the studies on resource allocation
for cellular networks focus on the single-cell setup, while the ICI
experienced by a receiver in one cell caused by the transmitters in
other cells is minimized by means of frequency reuse, which avoids
the same frequency band to be used by adjacent cells. However, most
beyond-3G wireless systems advocate to increase the frequency reuse
factor and even allow it to be one or so-called ``universal
frequency reuse'', due to which the issue of ICI becomes more
crucial. Consequently, joint resource allocation across neighboring
cells becomes a practically appealing approach for managing the ICI.
If the mobile stations (MSs) in each cell are separated for
transmission in frequency via orthogonal frequency-division
multiple-access (OFDMA) or in time via time-division multiple-access
(TDMA), then the active links in different cells transmitting at the
same frequency tone or in the same time slot will interfere with
each other, which can be modeled by a GIC. More specifically, if the
base stations (BS) and MSs are each equipped with one single
antenna, the system can be modeled as the single-input single-output
(SISO) GIC, termed as SISO-IC. If the BSs are each equipped with
multiple antennas while MSs are each equipped with one single
antenna, then in the uplink the system can be modeled as the
single-input multiple-output (SIMO) GIC, termed as SIMO-IC, and in
the downlink as the multiple-input single-input (MISO) GIC, termed
as MISO-IC.\footnote{It is worth noting that even for the
traditional single-cell setup with space-division multiple-access
(SDMA), i.e., the multi-antenna BS simultaneously communicating with
more than one single-antenna MSs, the MISO-IC and SIMO-IC models are
also applicable if the linear transmit/receive
precoding/equalization is implemented at the BS.}

The achievable rate region of SISO-IC, SIMO-IC or MISO-IC, with the
single-user detection (SUD) by treating the interference as Gaussian
noise, is in general a non-convex set due to the coupled
interference among users. As a result, how to efficiently find the
optimal power control and/or beamforming solutions to achieve the
maximum WSR for different types of GICs is a challenging problem. It
is worth noting that a great deal of valuable scholarly work
\cite{Gjendemsjoe08}-\cite{Schimidt} has contributed to resolving
this problem. For SISO-IC, various efficient power control schemes
have been studied. The authors in \cite{Gjendemsjoe08} showed that
in the two-user case the optimal power allocation to the sum-rate
maximization problem is ``binary'', i.e., either one user transmits
with full power and the other user shuts down, or both users
transmit with full power. However, this result does not hold in
general when the number of users is greater than two. Based on game
theory, an ``asynchronous distributed pricing (ADP)'' algorithm was
proposed in \cite{Huang06} whereby locally optimal solutions can be
obtained for WSR maximization. In \cite{Chiang07}, the WSR
maximization problem was transformed into a signomial programming
(SP) problem, which was efficiently solved by constructing a series
of geometric programming (GP) problems through the approach of
successive convex approximation. Similar to ADP, this algorithm only
guarantees locally optimal solutions. As for the case of parallel
SISO-IC, the authors in \cite{Yu06}, \cite{Luo08} showed that the
duality gap for the WSR maximization problem is zero when the number
of parallel GICs becomes asymptotically large. As a result, the
Lagrange duality method can be applied to decouple the problem into
parallel sub-problems in the dual domain. However, the power
optimization in each sub-problem for a given GIC is still
non-convex. For an extensive survey of power control algorithms for
SISO-IC, please refer to \cite{Chiang08}. Furthermore, for MISO-IC,
the optimality of transmit beamforming for achieving the maximum WSR
with SUD has been proven in \cite{Rui10,Shang09}. In
\cite{Rui10,Jorswieck08}, the complete characterization of all
Pareto optimal rates for MISO-IC was studied. To maximize the WSR,
an iterative algorithm was proposed in \cite{Jorswieck} from an
egotistic versus altruistic viewpoint, and other ``price-based''
algorithms (see, e.g., \cite{Schimidt} and references therein) were
also developed. However, these algorithms in general cannot achieve
the global WSR maximum for MISO-IC.

Different from the above prior work in which the power and/or
beamforming vectors were optimized directly for WSR maximization in
the GIC, in this paper we propose a new approach that maximizes the
WSR in the achievable rate region of the GIC directly. This approach
is based on the following two key observations: $(1)$ the WSR is a
strictly increasing function with respect to individual user rates;
and $(2)$ the achievable rate region is a ``normal'' set
\cite{Rubinov01}. Accordingly, the WSR maximization problem for the
GIC belongs to the class of optimization problems so-called
monotonic optimization over a normal set, for which the global
optimality can be achieved by an iterative ``outer polyblock
approximation'' algorithm \cite{Rubinov01}. However, one challenging
requirement of this algorithm is a unique characterization of the
Pareto boundary of the achievable rate region since at each
iteration of the algorithm one particular point on the Pareto
boundary that corresponds to the maximum achievable sum-rate in a
prescribed direction needs to be determined. This problem is
efficiently solved in this paper by utilizing a so-called ``rate
profile'' approach \cite{Mohseni06}, which transforms the original
problem into a sequence of signal-to-interference-plus-noise ratio
(SINR) feasibility problems. It is also shown in this paper that
such feasibility problems can be efficiently solved by existing
techniques for various types of GICs.

It is worth noting that rate profile was first proposed in
\cite{Mohseni06} as an alternative method to WSR maximization for
characterizing the Pareto boundary of the capacity region for the
multi-antenna Gaussian multiple-access channel (MAC), which is a
convex set. This method was later applied to characterize the Pareto
boundary of non-convex rate regions for the MISO-IC in \cite{Rui10}
and the two-way multi-antenna relay channel in \cite{Rui09}, for
which the WSR maximization approach is not directly applicable. A
very similar idea to rate profile was also proposed in
\cite{Shen05}, where the proportional rate fairness is imposed as a
constraint for WSR maximization in multi-user OFDM systems. As for
the outer polyblock approximation algorithm, it was first proposed
in \cite{Rubinov01}, and later applied in \cite{Ping09} and
\cite{Jorswieck10} to solve the WSR maximization problems for the
GIC. In \cite{Ping09}, this algorithm was applied for SISO-IC
together with the generalized linear fractional programming, which,
however, cannot be extended to SIMO-IC or MISO-IC. In
\cite{Jorswieck10}, this algorithm was applied to the two-user
MISO-IC by exploiting a prior result in \cite{Jorswieck08} that the
optimal transmit beamforming vector to achieve any Pareto boundary
rate-pair can be expressed as a linear combination of the
zero-forcing (ZF) and maximum-ratio transmission (MRT) beamformers.
However, this result only holds for the two-user MISO-IC and thus
how to extend the algorithm in \cite{Jorswieck10} to MISO-IC with
more than two users remains unknown. In comparison, in this paper we
show that by jointly utilizing the outer polyblock approximation
algorithm and rate profile approach, the global optimality of the
WSR maximization problem can be achieved for all SISO-IC, SIMO-IC
and MISO-IC, with arbitrary number of users.

It is also worth noting that for the WSR maximization in SISO-IC,
besides \cite{Ping09} that applies the outer polyblock approximation
algorithm, there have been other algorithms developed based on the
branch and bound method. For example, in \cite{Xu08} and
\cite{Weber10}, branch and bound methods combined with difference of
convex functions (DC) programming have been proposed. A generalized
branch and bound method applicable to problems in which the
objective function cannot be expressed in the form of DC, has also
been proposed in \cite{Weeraddana11}. In this paper, we propose an
alternative approach to that in the above prior work, whereby the
WSR maximization problems for SISO-IC, SIMO-IC, and MISO-IC are all
solvable.

The rest of this paper is organized as follows. Section
\ref{sec:system model and problem formulation} introduces the system
models for various GICs including SISO-IC, SIMO-IC and MISO-IC, and
formulates their WSR maximization problems. Section
\ref{sec:proposed solution} presents a new framework to solve the
formulated problems based on monotonic optimization and rate profile
techniques. Section \ref{sec:SINR feasibility problem} completes the
proposed algorithms by addressing the solutions to various SINR
feasibility problems. Section \ref{sec:numerical results} provides
numerical examples to validate the proposed results. Finally,
Section \ref{sec:conclusion} concludes the paper.

{\it Notation}: Scalars are denoted by lower-case letters, vectors
denoted by bold-face lower-case letters, and matrices denoted by
bold-face upper-case letters. $\mv{I}$ and $\mv{0}$  denote an
identity matrix and an all-zero matrix, respectively, with
appropriate dimensions. For a square matrix $\mv{S}$, $\mv{S}^{-1}$
denotes its inverse (if $\mv{S}$ is full-rank). For a matrix
$\mv{M}$ of arbitrary size, $\mv{M}^{H}$ and $\mv{M}^{T}$ denote the
conjugate transpose and transpose of $\mv{M}$, respectively. ${\rm
Diag}(\mv{X}_1,\cdots,\mv{X}_K)$ denotes a block diagonal matrix
with the diagonal matrices given by $\mv{X}_1,\cdots,\mv{X}_K$. The
distribution of a CSCG random vector with mean vector $\mv{x}$ and
covariance matrix $\mv{\Sigma}$ is denoted by
$\mathcal{CN}(\mv{x},\mv{\Sigma})$; and $\sim$ stands for
``distributed as''. $\mathbb{C}^{x \times y}$ denotes the space of
$x\times y$ complex matrices. $\mathbb{R}$ denotes the real number
space, while $\mathbb{R}^{x}$ denotes the $x \times 1$ real vector
space and $\mathbb{R}^{x}_+$ denotes its non-negative orthants.
$\|\mv{x}\|$ denotes the Euclidean norm of a complex vector
$\mv{x}$. $\mv{e}_k$ denotes a vector with its $k$th component being
$1$, and all other components being $0$. For two real vectors
$\mv{x}$ and $\mv{y}$, $\mv{x}\geq \mv{y}$ means that $\mv{x}$ is
greater than or equal to $\mv{y}$ in a component-wise manner.

\section{System Model and Problem Formulation}\label{sec:system
model and problem formulation}

\begin{figure}
\begin{center}
\scalebox{0.4}{\includegraphics*[40pt,309pt][449pt,780pt]{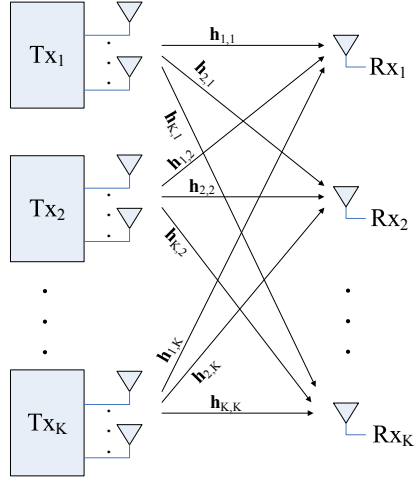}}
\end{center}
\caption{System model for the $K$-user MISO-IC (SISO-IC if each
transmitter has one single antenna, or SIMO-IC in the reverse link
transmission).}\label{fig1}
\end{figure}

We consider a $K$-user GIC, in which $K$ mutually interfering
wireless links communicate simultaneously over a common bandwidth,
as shown in Fig. \ref{fig1}. Firstly, consider the case where all
transmitters and receivers are each equipped with one single
antenna. The system is thus modeled as SISO-IC, for which the
discrete-time baseband signal received at the $k$th receiver is
given by\begin{align}\label{equ:received
signal}y_k=h_{k,k}\sqrt{p_k}x_k+\sum\limits_{j\neq
k}h_{k,j}\sqrt{p_j}x_j+z_k, ~~ k=1,\cdots,K,\end{align}where
$h_{k,j}$ denotes the complex channel gain from the $j$th
transmitter to the $k$th receiver, $p_k$ denotes the transmit power
of the $k$th transmitter, $x_k$ denotes the transmitted signal from
the $k$th transmitter, and $z_k$ denotes the background noise at the
$k$th receiver. It is assumed that
$z_k\sim\mathcal{CN}(0,\sigma_k^2)$, $\forall k$, and all $z_k$'s
are independent.

We assume independent encoding across different transmitters and
thus $x_k$'s are independent over $k$. It is also assumed that the
Gaussian codebook is used and thus $x_k\sim\mathcal{CN}(0,1)$.
Accordingly, the SINR of the $k$th receiver is expressed
as\begin{align}\label{equ:SINR for SISO}\gamma_k^{\rm
SISO-IC}=\frac{\|h_{k,k}\|^2p_k}{\sum\limits_{j\neq
k}\|h_{k,j}\|^2p_j+\sigma_k^2}.\end{align}

\begin{remark}It is worth noting that in the above signal
model, we have made the following two assumptions:
\begin{itemize}
\item[A1.] The interference is treated as additive Gaussian noise.
\item[A2.] The Gaussian input $x_k$ for user $k$ is assumed to be CSCG distributed and independent over time, i.e., asymmetric Gaussian signalling with time-domain symbol expansion in \cite{Jafar09} is not used.
\end{itemize}Note that for the subsequent studies on SIMO-IC and MISO-IC, the above two
assumptions are similarly made.
\end{remark}

Secondly, consider the case where all transmitters are each equipped
with one single antenna but each receiver is equipped with multiple
antennas, i.e., SIMO-IC. Assuming that the $k$th receiver is
equipped with $M_k>1$ antennas, its discrete-time baseband received
signal is given by\begin{align}\label{equ:received signal
uplink}y_k=\mv{w}_k^H(\mv{h}_{k,k}\sqrt{p_k}x_k+\sum\limits_{j\neq
k}\mv{h}_{k,j}\sqrt{p_j}x_j+\mv{z}_k), ~~
k=1,\cdots,K,\end{align}where $\mv{w}_k^H\in\mathbb{C}^{1\times
M_k}$ is the receive beamforming vector for the $k$th receiver,
$\mv{h}_{k,j}\in\mathbb{C}^{M_k\times 1}$ is the channel vector from
the $j$th transmitter to the $k$th receiver, and $\mv{z}_k\in
\mathbb{C}^{M_k\times 1}$ is the noise vector at the $k$th receiver.
It is assumed that $\mv{z}_k\sim
\mathcal{CN}(\mv{0},\sigma_k^2\mv{I})$. Thus, the SINR of the $k$th
receiver can be expressed as\begin{align}\label{equ:SINR in
SIMO-IC}\gamma_k^{\rm SIMO-IC}=\frac{p_k\|\mv{w}_k^H\mv{
h}_{k,k}\|^2}{\mv{w}_k^H(\sum\limits_{j\neq k}p_j\mv{h}_{k,j}\mv{
h}_{k,j}^H+\sigma_k^2\mv{I})\mv{w}_k}.\end{align}

Thirdly, consider the MISO-IC case in which all transmitters are
each equipped with multiple antennas while each receiver is equipped
with one single antenna. Assume that the $k$th transmitter is
equipped with $N_k>1$ antennas. The discrete-time baseband signal at
the $k$th receiver is then given by\begin{align}\label{equ:received
signal downlink}y_k=\mv{h}_{k,k}^H\mv{v}_kx_k+\sum\limits_{j\neq
k}\mv{h}_{k,j}^H\mv{v}_jx_j+z_k, ~~ k=1,\cdots,K,\end{align}where
$\mv{v}_k\in\mathbb{C}^{N_k\times 1}$ is the transmit beamforming
vector at the $k$th transmitter, and
$\mv{h}_{k,j}^H\in\mathbb{C}^{1\times N_j}$ denotes the channel
vector from the $j$th transmitter to the $k$th receiver.
Accordingly, the SINR of the $k$th receiver can be expressed
as\begin{align}\label{equ:SINR for MISO beamforming}\gamma_k^{\rm
MISO-IC}=\frac{\|\mv{h}_{k,k}^H \mv{v}_k\|^2}{\sum\limits_{j \neq
k}\|\mv{h}_{k,j}^H \mv{ v}_j\|^2+\sigma_k^2}.\end{align}

\begin{remark}
It is worth noting that for the above signal model for MISO-IC, we
assume that all transmitters employ rank-one beamforming. This is
because it has been shown in \cite{Rui10} and \cite{Shang09} that
under Assumptions A1 and A2, beamforming achieves all the points on
the Pareto boundary of the achievable rate region for MISO-IC, i.e.,
beamforming is optimal for WSR maximization in MISO-IC.
\end{remark}

With $\gamma_k$ defined in (\ref{equ:SINR for SISO}), (\ref{equ:SINR
in SIMO-IC}) or (\ref{equ:SINR for MISO beamforming}), the
achievable rate of the $k$th receiver can be formulated as
\begin{align}\label{equ:achievable rate}R_k(\gamma_k)=\log_2(1+\gamma_k), ~~ k=1,\cdots,K.\end{align}

Next, we define the achievable rate region for each type of GIC,
which constitutes all the rate-tuples simultaneously achievable by
all the users under a given set of transmit-power constraints
denoted by $P_1^{{\rm max}},\cdots,P_K^{{\rm max}}$:\begin{eqnarray}
&& \mathcal{R}^{{\rm
SISO-IC}}\triangleq\bigcup\limits_{\left\{p_k\right\}: 0\leq p_k\leq
P_k^{{\rm max}}, \ \forall k}\bigg \{ (r_1,\ldots,r_K):0\leq r_k\leq
R_k(\gamma_k^{{\rm SISO-IC}}), k=1,\ldots,K
\bigg\},\label{equ:rate region SISO} \\
&& \mathcal{R}^{{\rm SIMO-IC}}
\triangleq\bigcup\limits_{\left\{p_k\right\},\{\mv{w}_k\}: 0\leq
p_k\leq P_k^{{\rm max}}, \ \forall k}\bigg \{ (r_1,\ldots,r_K):0\leq
r_k\leq R_k(\gamma_k^{{\rm SIMO-IC}}), k=1,\ldots,K \bigg\},
\label{equ:rate region SIMO}
\\ && \mathcal{R}^{{\rm
MISO-IC}}\triangleq\bigcup\limits_{\left\{\mv{v}_k\right\}:0\leq
\|\mv{v}_k\|^2\leq P_k^{{\rm max}}, \ \forall k}\bigg \{
(r_1,\ldots,r_K):0\leq r_k\leq R_k(\gamma_k^{{\rm MISO-IC}}),
k=1,\ldots,K \bigg\}.\label{equ:rate region MISO}
\end{eqnarray}The upper-right boundary of each defined rate region is called the {\it Pareto
boundary}, constituted by rate-tuples for each of which it is
impossible to improve one particular user's rate without decreasing
the rate of at least one of the other users.

The WSR maximization problems for SISO-IC, SIMO-IC and MISO-IC are
then formulated as (P1.1)-(P1.3) as follows.
\begin{align*}\mathrm{(P1.1)}:~\mathop{\mathtt{Maximize}}_{\mv{p}} & ~~~ U(\mv{p}):=\sum\limits_{k=1}^K\mu_kR_k(\gamma_k^{{\rm SISO-IC}}) \\
\mathtt {Subject \ to} & ~~~ 0\leq p_k\leq P_k^{{\rm max}}, ~
\forall k,
\end{align*}\vspace{-0.2in}
\begin{align*}\mathrm{(P1.2)}:~\mathop{\mathtt{Maximize}}_{\mv{W},\mv{p}} & ~~~ U(\mv{W},\mv{p}):=\sum\limits_{k=1}^K\mu_kR_k(\gamma_k^{{\rm SIMO-IC}}) \\
\mathtt {Subject \ to} & ~~~ 0\leq p_k\leq P_k^{{\rm max}}, ~
\forall k,
\end{align*}\vspace{-0.2in}
\begin{align*}\mathrm{(P1.3)}:~\mathop{\mathtt{Maximize}}_{\mv{V}} & ~~~ U(\mv{V}):=\sum\limits_{k=1}^K\mu_kR_k(\gamma_k^{{\rm MISO-IC}}) \\
\mathtt {Subject \ to} & ~~~ \|\mv{v}_k\|^2\leq P_k^{{\rm max}},
\forall k,
\end{align*}where $\mv{p}=(p_1,\cdots,p_K)$ denotes the transmit power
vector, $\mv{W}=(\mv{w}_1,\cdots,\mv{w}_K)$ and
$\mv{V}=(\mv{v}_1,\cdots,\mv{v}_K)$ constitute the receive and
transmit beamforming vectors, respectively, and $\mu_k$ is the
non-negative rate weight for user $k$. Since the objective functions
are all non-concave with respect to the power values or beamforming
vectors due to the coupled interference, all the WSR maximization
problems in (P1.1)-(P1.3) are non-convex and thus cannot be solved
globally optimally by conventional convex optimization techniques.

\section{Proposed Solution Based on Outer Polyblock Approximation and Rate
Profile}\label{sec:proposed solution}

In this section, we solve the formulated WSR maximization problems
in (P1.1)-(P1.3) globally optimally by a new approach based on the
outer polyblock approximation and rate profile techniques.

\subsection{A New Look at the Problem: Optimizing WSR Directly in Rate Region}

Traditionally, Problems (P1.1)-(P1.3) are solved in the power
allocation and/or beamforming domain, which results in non-convex
optimization problems. In this subsection, we study the WSR
maximization problem utilizing a new formulation, which maximizes
the WSR directly in the achievable rate region.

If the achievable rate vector $\mv{r}=(R_1,\cdots,R_K)$ is treated
as the design variable, where $R_k$ is the achievable rate of user
$k$ defined in (\ref{equ:achievable rate}), the WSR maximization
problems (P1.1)-(P1.3) can be unified in the following
form.\begin{align*}\mathrm{(P2)}:~\mathop{\mathtt{Maximize}}_{\mv{r}} & ~~~ U(\mv{r}):=\sum\limits_{k=1}^K\mu_kR_k \\
\mathtt {Subject \ to} & ~~~ \mv{r}\in \mathcal{R},
\end{align*}where the
rate region $\mathcal{R}$ is defined in (\ref{equ:rate region
SISO}), (\ref{equ:rate region SIMO}) or (\ref{equ:rate region MISO})
for SISO-IC, SIMO-IC or MISO-IC.

Next, we will show that Problem (P2) belongs to one special class of
optimization problems: monotonic optimization over a ``normal'' set.
Two useful definitions are given first as follows.

\begin{definition}A function $f:\mathbb{R}^n\rightarrow\mathbb{R}$
is said to be strictly increasing on $\mathbb{R}_+^n$ if for any
$\mv{x}',\mv{x}\in \mathbb{R}_+^n$, $\mv{x}'\geq \mv{x}$ and
$\mv{x}'\neq \mv{x}$ imply that $f(\mv{x}')>
f(\mv{x})$.\end{definition}

\begin{definition}A set $\mathcal{D}\in \mathbb{R}_+^n$ is called normal if
given any point $\mv{x}\in \mathcal{D}$, all the points $\mv{x}'$
with $\mv{0}\leq \mv{x}'\leq \mv{x}$ satisfy that $\mv{x}'\in
\mathcal{D}$.\end{definition}

Based on the above definitions, we declare the following two facts
regarding Problem (P2), which can be easily verified to be true.

\begin{fact}\label{fact:objective} The objective function of Problem (P2) is a
strictly increasing function with respect to $\mv{r}$.
\end{fact}

\begin{fact}\label{fact:constraint} The achievable rate region defined in (\ref{equ:rate region
SISO}), (\ref{equ:rate region SIMO}) or (\ref{equ:rate region MISO})
is a normal set.
\end{fact}

Facts \ref{fact:objective} and \ref{fact:constraint} imply that
Problem (P2) maximizes a strictly increasing function over a normal
set. In \cite{Rubinov01}, the ``outer polyblock approximation''
algorithm was proposed to achieve the global optimality for this
type of problems. In the following, we will apply this algorithm to
solve Problem (P2).

\subsection{Outer Polyblock Approximation
Algorithm}\label{subsec:algorithm}

In this subsection, we introduce the outer polyblock approximation
algorithm to solve Problem (P2). First, two definitions are given as
follows.

\begin{definition}Given any vector $\mv{v}\in\mathbb{R}_+^n$, the
hyper rectangle $[\mv{0},\mv{v}]=\{\mv{x}|\mv{0}\leq \mv{x} \leq
\mv{v}\}$ is referred to as a box with vertex $\mv{v}$.
\end{definition}

\begin{definition}A set is called a polyblock if it is the union of
a finite number of boxes.
\end{definition}

Next, we show one important property of the polyblock in the
following proposition.

\begin{proposition}\label{pro:optimality}The maximum of
a strictly increasing function $f(\mv{x})$ over a polyblock is
achieved at one of the vertices of the polyblock.
\end{proposition}

\begin{proof}
Suppose that $\mv{x}^\ast$ is the globally optimal solution over the
polyblock, and it is not a vertex. Then, there exists at least one
vertex $\mv{x}'$ satisfying $\mv{x}'\geq \mv{x}^\ast$ but
$\mv{x}'\neq \mv{x}^\ast$. Since $f(\mv{x})$ is a strictly
increasing function, $f(\mv{x}^\ast)<f(\mv{x}')$ must hold, which
contradicts to the presumption. The proof is thus completed.
\end{proof}

According to Proposition \ref{pro:optimality}, the maximum of an
increasing function over a polyblock can be obtained efficiently by
enumeration of the vertices of that polyblock. Consequently, we can
construct a sequence of polyblocks to approximate the rate region
$\mathcal{R}$ with the increasing accuracy for Problem (P2). In
other words, we need to find an iterative method to generate a
sequence of polyblocks of shrinking sizes such that
\begin{eqnarray}& P^{(1)}\supset P^{(2)}\supset\cdots\supset \mathcal{R}, \label{equ:contain} \\ & \lim\limits_{n \to \infty}[\max\limits_{\mv{r}\in P^{(n)}}U(\mv{r})]=\max \limits_{\mv{r} \in \mathcal{R}}U(\mv{r}),  \label{equ:optimality}\end{eqnarray}where
$P^{(n)}$ denotes the polyblock generated at the $n$th iteration.

Next, we present one method to generate the polyblocks satisfying
(\ref{equ:contain}) and (\ref{equ:optimality}). Let
$\mathcal{Z}^{(n)}$ denote the set containing all the vertices of
the polyblock $P^{(n)}$. The vertex that achieves the maximum WSR in
polyblock $P^{(n)}$ can be formulated
as\begin{align}\label{equ:optimal vertex}\tilde{\mv{z}}^{(n)}=\arg
\max \limits_{\mv{z}\in \mathcal{Z}^{(n)}}U(\mv{z}).\end{align}

\begin{figure}
\begin{center}
\scalebox{0.7}{\includegraphics*[8pt,235pt][485pt,658pt]{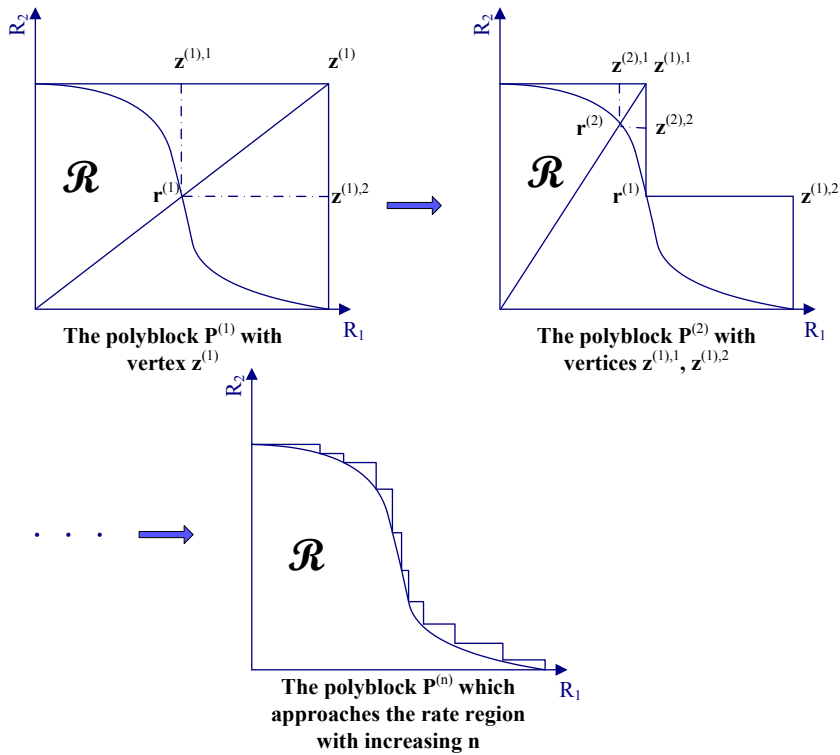}}
\end{center}
\caption{Illustration of the procedure for constructing new
polyblocks.}\label{fig2}
\end{figure}

Define $\delta \tilde{\mv{z}}^{(n)}$ as the line that connects the
two points $\mv{0}$ and $\tilde{\mv{z}}^{(n)}$, and $\mv{r}^{(n)}$
as the intersection point on the Pareto boundary with the line
$\delta \tilde{\mv{z}}^{(n)}$. The following method can be used to
generate $K$ new vertices adjacent to $\tilde{\mv{z}}^{(n)}$:
\begin{align}\label{equ:new vertex}
\mv{z}^{(n),i}=\tilde{\mv{z}}^{(n)}-(\tilde{z}_i^{(n)}-r_i^{(n)})\mv
e_i, ~~ i=1,\cdots,K,\end{align}where $\mv{z}^{(n),i}$ denotes the
$i$th new vertex generated at the $n$th iteration;
$\tilde{z}_i^{(n)}$ and $r_i^{(n)}$ denote the $i$th element of
vectors $\tilde{\mv{z}}^{(n)}$ and $\mv{r}^{(n)}$, respectively.
Then, the new vertex set can be expressed as
\begin{align}\label{equ:new
set}\mathcal{Z}^{(n+1)}=\mathcal{Z}^{(n)}\backslash
\tilde{\mv{z}}^{(n)}\bigcup \{\mv{z}^{(n),1},\cdots,\mv{
z}^{(n),K}\}.\end{align}Each vertex in the set $\mathcal{Z}^{(n+1)}$
defines a box, and thus the new polyblock $P^{(n+1)}$ is the union
of all these boxes. An illustration about the above procedure to
generate polyblocks for the case of two-user rate region is given in
Fig. \ref{fig2}. In the following proposition, we show the
feasibility of the above polyblock generation method.

\begin{proposition}\label{pro:shrinking}If the rate region $\mathcal{R}$ is a normal set (as we have already shown),
the polyblocks generated by (\ref{equ:new set}) satisfy
(\ref{equ:contain}).
\end{proposition}

\begin{proof}Please refer to \cite{Rubinov01}.
\end{proof}

Proposition \ref{pro:shrinking} ensures that the above polyblock
generation method can be used to approximate the rate region from
the outside with increasing accuracy. Let
$\mv{r}^\ast=(R_1^\ast,\cdots,R_K^\ast)$ denote the optimal solution
to Problem (P2). Based on the above method, in the following we
present an algorithm to find $\mv{r}^\ast$ in the rate region
$\mathcal{R}$. It is worth noting that $\mv{r}^\ast$ must be on the
Pareto boundary of the rate region; thus, we only need to search
over the Pareto boundary to find $\mv{r}^\ast$.

The outer polyblock approximation algorithm works iteratively as
follows. In the $n$th iteration, the optimal vertex
$\tilde{\mv{z}}^{(n)}$ is first obtained by (\ref{equ:optimal
vertex}). According to Proposition \ref{pro:optimality}, in the
polyblock $P^{(n)}$ the maximum WSR is $U(\tilde{\mv{z}}^{(n)})$.
Since Proposition \ref{pro:shrinking} implies that $P^{(n)}$ always
contains the rate region $\mathcal{R}$, $U(\tilde{\mv{z}}^{(n)})$ is
an upper bound of $U(\mv{r}^\ast)$. Then, the intersection point
$\mv{r}^{(n)}$ on the Pareto boundary with the line $\delta
\tilde{\mv{z}}^{(n)}$ is obtained. Define the best intersection
point up to the $n$th iteration as\begin{align}\label{equ:best
intersection}\tilde{\mv{r}}^{(n)}=\arg
\max\{U(\mv{r}^{(n)}),U(\tilde{\mv{r}}^{(n-1)})\}.\end{align}Consequently,
$U(\tilde{\mv{r}}^{(n)})$ is the tightest lower bound of
$U(\mv{r}^\ast)$ by the $n$th iteration. Next, we can compute the
value of $U(\tilde{\mv{z}}^{(n)})-U(\tilde{\mv{r}}^{(n)})$, which is
the difference between the upper and lower bounds of the optimal
value of Problem (P2) at the $n$th iteration. If this difference is
less than $\eta$ (a small positive number), the algorithm can
terminate and $\tilde{\mv{r}}^{(n)}$ is at least an $\eta$-optimal
solution to Problem (P2) because
\begin{align}\label{equ:stopping criterion}U(\mv{r}^\ast)
-U(\tilde{\mv{r}}^{(n)})<U(\tilde{\mv{z}}^{(n)})-U(\tilde{\mv{r}}^{(n)})<
\eta.\end{align}Otherwise, we construct a new polyblock $P^{(n+1)}$
by the above polyblock generation method. We repeat the above
procedure until an $\eta$-optimal solution is found.

\begin{table}[ht]
\begin{center}
\caption{\textbf{Algorithm \ref{table1}}: Outer Polyblock
Approximation Algorithm for Solving Problem (P2)} \vspace{0.2cm}
 \hrule
\vspace{0.3cm}
\begin{itemize}
\item[a)] Initialize: Set $n=1$,
$\mathcal{Z}^{(1)}=\{\mv{z}^{(1)}\}$;
\item[b)] {\bf While} $(\epsilon,\eta)$-accuracy is not reached, {\bf do}
\begin{itemize}
\item[1)] Find the optimal vertex $\tilde{\mv z}^{(n)}$ that maximizes the
WSR in the set $\mathcal{Z}^{(n)}_\epsilon$ based on
\begin{align}\label{equ:optimal vertex1}\tilde{\mv{z}}^{(n)}=\arg
\max \limits_{\mv{z}\in
\mathcal{Z}^{(n)}_\epsilon}U(\mv{z}),\end{align}where $\epsilon$ is
a small positive number and $\mathcal{Z}^{(n)}_\epsilon=\{\mv{z}\in
\mathcal{Z}^{(n)}|z_k\geq \epsilon, \ \forall k \}$;
\item[2)] Compute the intersection point $\mv{r}^{(n)}$ on the Pareto boundary of the rate region $\mathcal{R}$ with the line $\delta \tilde{\mv{z}}^{(n)}$ ;
\item[3)] Update the best intersection point until the $n$th iteration $\tilde{\mv{r}}^{(n)}$ according to (\ref{equ:best
intersection});
\item[4)] {\bf If} $U(\tilde{\mv{z}}^{(n)})-U(\tilde{\mv{r}}^{(n)})\leq \eta$, {\bf then}
\begin{itemize}
\item[] Stop and $\tilde{\mv{r}}^{(n)}$ is an $(\epsilon,\eta)$-optimal solution
to Problem (P2);
\end{itemize}
\item[5)] {\bf else}
\begin{itemize}
\item[] Compute $K$ new vertices that are adjacent to $\tilde{\mv{z}}^{(n)}$ by
(\ref{equ:new vertex}) and update the vertex set
$\mathcal{Z}^{(n+1)}$ by (\ref{equ:new set});
\end{itemize}
\item[6)] {\bf end}
\item[7)] $n=n+1$;
\end{itemize}
\item[c)] {\bf end}
\end{itemize}
\vspace{0.2cm} \hrule \label{table1} \end{center}
\end{table}

The above algorithm, denoted as Algorithm \ref{table1}, is
summarized in Table \ref{table1}. It is worth noting that in
Algorithm \ref{table1}, $\tilde{\mv{z}}^{(n)}$ is obtained by
enumeration in the set $\mathcal{Z}^{(n)}_\epsilon$ rather than
$\mathcal{Z}^{(n)}$. This is because in \cite{Rubinov01} it was
shown that if the optimal solution lies in a strip defined by
$\{\mv{r}^\ast|0\leq R_k^\ast \leq \epsilon\}$ with arbitrary $k$
and a small value $\epsilon>0$, then as $\tilde{\mv{z}}^{(n)}$
approaches this strip, the algorithm converges very slowly.
Consequently, $\epsilon$ is chosen to balance the tradeoff between
the accuracy and complexity of Algorithm \ref{table1}. With
$\epsilon$, Algorithm \ref{table1}
solves the following problem\begin{align*}\mathrm{(P2-A)}:~\mathop{\mathtt{Maximize}}_{\mv{r}} & ~~~ U(\mv{r}):=\sum\limits_{k=1}^K\mu_kR_k \\
\mathtt {Subject \ to} & ~~~ \mv{r}\in \mathcal{R}^\epsilon,
\end{align*}where $\mathcal{R}^\epsilon$ is defined
as\begin{align}\mathcal{R}^\epsilon=\mathcal{R}\cap
\{(r_1,\cdots,r_K):r_k\geq \epsilon, \forall k\}.\end{align}Thus,
the corresponding solution is called an $(\epsilon,\eta)$-optimal
solution to Problem P2.

Next, we address the convergence issue of Algorithm \ref{table1}.
According to Proposition \ref{pro:shrinking}, $P^{(n)}\supset
P^{(n+1)}$ always holds. Moreover, the optimal vertex
$\tilde{\mv{z}}^{(n)}$ is removed from
$\mathcal{Z}^{(n+1)}_\epsilon$ after each iteration. Thus,
$U(\tilde{\mv{z}}^{(n+1)})<U(\tilde{\mv{z}}^{(n)})$ also holds.
Furthermore, the lower bound $U(\tilde{\mv{r}}^{(n)})$ is
non-decreasing. Consequently, the value of
$U(\tilde{\mv{z}}^{(n)})-U(\tilde{\mv{r}}^{(n)})$ will decrease
after each iteration. It was shown in \cite{Rubinov01} that as $n$
increases, the difference between the upper and lower bounds can be
reduced to an arbitrary small positive number in a finite number of
iterations. Thus, Algorithm \ref{table1} converges given small
positive values $\epsilon$ and $\eta$. More details about the
selection of the values of $\epsilon$ and $\eta$ will be given later
in Section \ref{subsec:convergence performance}.

Last, we explain how to obtain an initial vertex
$\mv{z}^{(1)}=(z_1^{(1)},\cdots,z_K^{(1)})$ for the first iteration
of Algorithm \ref{table1}. Since the box $[\mv{0},\mv{z}^{(1)}]$
needs to contain the rate region $\mathcal{R}$, for any user $k$,
$z_k^{(1)}$ can be obtained when all other users switch off their
transmission (thus no interference exists for user $k$), and user
$k$ transmits its maximum power $P_k^{{\rm max}}$. More
specifically, for SISO-IC,
\begin{align}\label{equ:vertex}z_k^{(1)}=\log_2(1+\frac{P_k^{{\rm
max}}\|h_{k,k}\|^2}{\sigma_k^2}), ~ \forall k.
\end{align}

Since for MISO-IC,
\begin{align}
\gamma_k^{{\rm MISO-IC}} & =\frac{\|\mv{
h}_{k,k}^H\mv{w}_k\|^2}{\sum\limits_{j\neq
k}\|\mv{h}_{k,j}^H\mv{w}_j\|^2+\sigma_k^2}
<\frac{\|\mv{h}_{k,k}^H\mv{w}_k\|^2}{\sigma_k^2} \nonumber
\\ & \overset{(a)}{\leq}
\frac{\|\mv{w}_k\|^2\|\mv{h}_{k,k}\|^2}{\sigma_k^2}\leq
\frac{P_k^{{\rm max}}\|\mv{h}_{k,k}\|^2}{\sigma_k^2}, \ \forall
k,\label{equ:SINR in MISO-BC relaxation}\end{align}where (a) is due
to the Cauchy-Schwarz inequality, $z_k^{(1)}$ can thus be set as
\begin{align}\label{equ:first vertex MISO-IC}z_k^{(1)}=\log_2(1+\frac{P_k^{{\rm max}}\|\mv{
h}_{k,k}\|^2}{\sigma_k^2}), ~ \forall k.\end{align}

The initial vertex for SIMO-IC can be obtained similarly to
(\ref{equ:first vertex MISO-IC}), and is thus omitted for brevity.

To summarize, the only challenge that remains unaddressed in
Algorithm \ref{table1} is on how to compute the intersection point
$\mv{r}^{(n)}$ on the Pareto rate boundary with the line $\delta
\tilde{\mv{z}}^{(n)}$ at the $n$th iteration, which will be
addressed next.

\subsection{Finding Intersection Points by the ``Rate Profile'' Approach}

In this subsection, we show how to obtain the intersection point on
the Pareto boundary of the rate region with the line $\delta
\tilde{\mv{z}}^{(n)}$, to complete Algorithm \ref{table1}. Let
$\footnotesize{R_{{\rm sum}}=\sum\limits_{k=1}^{K}R_k}$ denote the
sum-rate of all the users, $\footnotesize{\mv
\alpha=\tilde{\mv{z}}^{(n)}/\sum\limits_{k=1}^K\tilde{z}_k^{(n)}}$
denote the slope of the line $\delta \tilde{\mv{z}}^{(n)}$.
Consequently, the intersection point at the $n$th iteration can be
expressed as $\mv{r}^{(n)}=R_{{\rm sum}}^\ast \mv \alpha$, where
$R_{{\rm sum}}^\ast$ is the optimal value of the following
problem:\begin{align}\label{equ:intersection point in one
direction}\mathop{\mathtt{Maximize}} & ~~~ R_{{\rm sum}} \nonumber\\
\mathtt{Subject \ to} &~ ~~ R_{{\rm sum}} \mv \alpha \in
\mathcal{R}.\end{align}

The above approach to find the intersection point on the Pareto
boundary of the rate region is known as {\it rate profile}
\cite{Rui10,Mohseni06,Rui09}. In the following, we solve Problem
(\ref{equ:intersection point in one direction}) to obtain the
intersection point $\mv{r}^{(n)}$ on the Pareto boundary with a
given $\delta\tilde{\mv{z}}^{(n)}$.

Problem (\ref{equ:intersection point in one direction}) is solvable
via solving a sequence of feasibility problems shown as follows.
Given a target sum-rate $\bar{R}_{{\rm sum}}$, the feasibility
problems for SISO-IC, SIMO-IC and MISO-IC can be expressed in the
following problems (P3.1)-(P3.3), respectively.
\begin{align*}\mathrm{(P3.1)}:~\mathop{\mathtt{Find}} & ~~\{p_k\} \\
\mathtt{Subject \ to} &
~~\log_2(1+\gamma_k^{{\rm SISO-IC}})\geq \alpha_k\bar{R}_{{\rm sum}}, ~ \forall k \\
& ~~p_k\leq P_k^{{\rm max}}, ~~ \forall k.\end{align*}
\begin{align*}\mathrm{(P3.2)}:~\mathop{\mathtt{Find}} &
~~\{\mv{w}_k\}, ~ \{p_k\} \nonumber\\ \mathtt{Subject \ to} &
~~\log_2(1+\gamma_k^{{\rm SIMO-IC}})\geq \alpha_k\bar{R}_{{\rm sum}}, ~ \forall k \\
& ~~p_k\leq P_k^{{\rm max}}, ~~ \forall k.\end{align*}
\begin{align*}\mathrm{(P3.3)}:~\mathop{\mathtt{Find}} & ~~\{\mv{v}_k\}
\nonumber\\ \mathtt{Subject \ to} &
~~\log_2(1+\gamma_k^{{\rm MISO-IC}})\geq \alpha_k\bar{R}_{{\rm sum}}, ~ \forall k \\
& ~~\|\mv{v}_k\|^2\leq P_k^{{\rm max}}, ~ \forall k.
\end{align*}

If any of Problems (P3.1), (P3.2) and (P3.3) is feasible, it follows
that $R_{{\rm sum}}^\ast\geq \bar{R}_{{\rm sum}}$; otherwise,
$R_{{\rm sum}}^\ast < \bar{R}_{{\rm sum}}$. Hence, $R_{{\rm
sum}}^\ast$ can be obtained for Problem (\ref{equ:intersection point
in one direction}) by applying a simple bisection method
\cite{Boyd04}, for which the detail is omitted for brevity.

The remaining challenge is on solving the feasibility problems
(P3.1)-(P3.3), which is addressed next. Let
$\bar{\gamma}_k=2^{\alpha_k \bar{R}_{{\rm sum}}}-1$, $\forall k$.
Then, the first constraint of each feasibility problem can be
re-expressed as
\begin{align}\label{equ:sinr feasibility}\gamma_k \geq \bar{\gamma}_k, ~~ \forall k.\end{align}Therefore, given any sum-rate target $\bar{R}_{{\rm sum}}$,
the feasibility problems (P3.1)-(P3.3) are equivalent to finding
whether a corresponding SINR target vector $\bar{\mv
\gamma}=(\bar{\gamma}_1,\cdots,\bar{\gamma}_K)$ is achievable. In
the next section, we will propose efficient algorithms to solve
these SINR feasibility problems.

\begin{remark}In the case where a set of minimum rate constraints $\label{equ:rate constraints}R_k\geq R_k^{{\rm min}}, ~
\forall k$, are added to the WSR maximization problem (P2), where
$R_k^{{\rm min}}$ is the minimum rate required for user $k$, we can
solve this new problem by modifying Algorithm \ref{table1} as
follows. Since the new rate region $\mathcal{R}'$ is the
intersection of the original rate region with the set
$\{(r_1,\cdots,r_K):r_k\geq R_k^{{\rm min}}, \ \forall k\}$, we
should change the initial point from $\mv{0}$ to $\mv{r}^{{\rm
min}}$ in Algorithm \ref{table1}, where $\mv{r}^{{\rm
min}}=(R_1^{{\rm min}},\cdots,R_K^{{\rm min}})$ is the rate
constraint vector. Thus, at each iteration we need to compute the
intersection point on the Pareto boundary with the line passing
through the optimal vertex $\tilde{\mv{z}}^{(n)}$ and the point
$\mv{r}^{{\rm min}}$ (instead of $\mv{0}$ in Algorithm
\ref{table1}). In addition, any point $\mv{r}$ on this line with
$R_{{\rm sum}}=\sum_{k=1}^KR_k$ can be rewritten
as\begin{align}\mv{r}=\mv{r}^{{\rm min}}+\mv \alpha(R_{{\rm
sum}}-\sum_{k=1}^KR_k^{{\rm min}}),\end{align}where the rate profile
$\mv \alpha$ is obtained by $\mv
\alpha=\frac{\tilde{\mv{z}}^{(n)}-\mv{r}^{{\rm
min}}}{\sum_{k=1}^K\tilde{z}_k^{(n)}-\sum_{k=1}^KR_k^{{\rm min}}}$.

\end{remark}

\section{Solutions to SINR Feasibility
Problems}\label{sec:SINR feasibility problem}

In this section, we solve Problems (P3.1)-(P3.3) subject to the
equivalent SINR constraints given in (\ref{equ:sinr feasibility})
for SISO-IC, SIMO-IC and MISO-IC, respectively.

\subsection{The
SISO-IC Case}\label{sec:SISO}

We first study the feasibility problem (P3.1) for SISO-IC. Given a
SINR target vector $\bar{\mv
\gamma}=(\bar{\gamma}_1,\cdots,\bar{\gamma}_K)$ with
$\bar{\gamma}_k=2^{\alpha_k \bar{R}_{{\rm sum}}}-1$, we will check
whether it is achievable under users' individual power constraints.

Let $\mv{G}$ denote the $K \times K$ normalized channel gain matrix
given by
\begin{align}G_{k,j}=\left\{
\begin{array}{ll} \frac{\bar{\gamma}_k\|h_{k,j}\|^2}{\|h_{k,k}\|^2}, &
k\neq j \\ 0, & k=j , \end{array} \right.\end{align}and $\mv \eta$
denote the $K \times 1$ normalized noise vector given
by\begin{equation}
\eta_k=\frac{\bar{\gamma}_k\sigma_k^2}{\|h_{k,k}\|^2}, ~~ \forall
k.\end{equation}To achieve the SINR target, the transmit power
vector for users is given by
\begin{equation}\label{equ:power control}\mv{p}=(\mv{I}-\mv{
G})^{-1}\mv \eta.\end{equation}

Let $\rho(\mv{B})$ denote the spectral radius (defined as the
maximum eigenvalue in absolute value) of the non-negative matrix
$\mv{B}$. The following propositions were shown in \cite{Bambos96},
which play important roles in solving Problem (P3.1).

\begin{proposition}\label{the:feasibility}The power allocation $\mv{p}$ given by (\ref{equ:power control}) satisfies $\mv{p}\geq \mv{0}$ if
and only if $\rho(\mv{G})<1$.
\end{proposition}

\begin{proposition}\label{the:power constraint}If $\rho(\mv{G})<1$, the power allocation $\mv{p}$ given by (\ref{equ:power control}) is component-wise minimum in the
sense that any other power allocation $\mv{p}'$ that satisfies
(\ref{equ:sinr feasibility}) needs to satisfy $\mv{p}'\geq \mv{p}$.
\end{proposition}

Propositions \ref{the:feasibility} and \ref{the:power constraint}
imply that a SINR target vector $\bar{\mv \gamma}$ is feasible if
and only if: $(1)$ $\rho(\mv{G})<1$, and $(2)$ the power solution
obtained by (\ref{equ:power control}) satisfies $p_k\leq P_k^{{\rm
max}}$, $\forall k$. Consequently, we propose Algorithm \ref{table2}
in Table \ref{table2} to solve Problem (P3.1).

\begin{table}[htp]
\begin{center}
\caption{\textbf{Algorithm \ref{table2}}: Algorithm for Solving
Problem (P3.1)} \vspace{0.2cm}
 \hrule
\vspace{0.2cm}
\begin{itemize}
\item[a)] Given any SINR target vector $\bar{\mv \gamma}=(\bar{\gamma}_1,\cdots,\bar{\gamma}_K)$, compute the spectrum radius of
matrix $\mv{G}$. If it is larger than $1$, conclude that there is no
feasible power allocation to meet the SINR target and exit the
algorithm; otherwise, go to step $b)$;
\item[b)] Compute the power allocation $\mv{p}$ by (\ref{equ:power
control}), and check for any user $k$, whether the power constraint
$p_k\leq P_k^{{\rm max}}$ is satisfied. If so, conclude that the
SINR target is feasible; otherwise, the SINR target is not feasible.
\end{itemize}
\vspace{0.2cm} \hrule \label{table2}
\end{center}
\end{table}

\begin{remark} It is worth comparing Algorithm \ref{table1} with the
MAPEL algorithm proposed in \cite{Ping09}. MAPEL solves Problem
(P1.1) for SISO-IC in the SINR region (as opposed to the rate region
in our approach) due to the fact that the problem to characterize
the Pareto boundary of the SINR region for SISO-IC can be
transformed into a generalized linear fractional programming problem
and thus efficiently solved by Dinkelbach-type algorithm
\cite{Frenk06}. However, this transformation does not work for
SIMO-IC or MISO-IC if the beamforming vectors are involved.
Consequently, MAPEL cannot be extended to the GIC with multiple
antennas. As comparison, in this paper we solve the WSR maximization
problem in the rate region directly because the Pareto boundary can
be characterized completely by the rate profile approach, as along
as the associated SINR feasibility problem can be solved. Thus, our
proposed algorithm is more applicable than MAPEL in solving the WSR
maximization problems for SIMO-IC and MISO-IC, as shown next.
\end{remark}

\subsection{The
SIMO-IC Case}\label{sec:SIMO-IC}

The feasibility of Problem (P3.2) can be checked by using the
optimal value of the following SINR balancing
problem:\begin{align}\label{equ:SINR balancing problem for SIMO-IC
with individual power constraint}\mathop{\mathtt{Maximize}}
& ~~ \min\limits_{1\leq k \leq K}\frac{\gamma_k}{\bar{\gamma}_k} \nonumber \\
\mathtt{Subject \ to} & ~~ p_k\leq P_k^{{\rm max}}, ~ \forall
k.\end{align}If the optimal value of Problem (\ref{equ:SINR
balancing problem for SIMO-IC with individual power constraint}) is
no smaller than $1$, then the SINR target vector $\bar{\mv
\gamma}=(\bar{\gamma}_1,\cdots,\bar{\gamma}_K)$ is achievable;
otherwise, this SINR target cannot be achieved.

In \cite{Schubert04}, an efficient algorithm was proposed to solve a
SINR balancing problem similar to Problem (\ref{equ:SINR balancing
problem for SIMO-IC with individual power constraint}), where only
one sum-power constraint is imposed. However, the algorithm in
\cite{Schubert04} cannot be directly applied to solve Problem
(\ref{equ:SINR balancing problem for SIMO-IC with individual power
constraint}) due to multiple users' individual power constraints. To
utilize the algorithm proposed in \cite{Schubert04}, we decouple
Problem (\ref{equ:SINR balancing problem for SIMO-IC with individual
power constraint}) into $K$ sub-problems, with the $i$th sub-problem
formulated
as:\begin{align}\label{equ:sub-problem}\mathop{\mathtt{Maximize}} &
~~ \min\limits_{1\leq k \leq K}\frac{\gamma_k}{\bar{\gamma}_k}
\nonumber \\ \mathtt{Subject \ to} & ~~ p_i\leq P_i^{{\rm max}}.
\end{align}Therefore, for the $i$th sub-problem only the $i$th
user's power constraint is considered. Next, we show how to solve
Problem (\ref{equ:sub-problem}) by extending the algorithm in
\cite{Schubert04}, and then reveal an important relationship between
Problems (\ref{equ:SINR balancing problem for SIMO-IC with
individual power constraint}) and (\ref{equ:sub-problem}), based
upon which we further propose an efficient algorithm to solve
Problem (\ref{equ:SINR balancing problem for SIMO-IC with individual
power constraint}).

\subsubsection{Solution to Problem (\ref{equ:sub-problem})}
\ \ \

In this part, we extend the algorithm proposed in \cite{Schubert04}
to solve Problem (\ref{equ:sub-problem}) for a given $i$.

One important property of the SINR balancing problem in
(\ref{equ:sub-problem}) is that given any receive beamforming
vectors $\bar{\mv{W}}=(\bar{\mv{w}}_1,\cdots,\bar{\mv{w}}_K)$, the
corresponding optimal power allocation $\bar{\mv{p}}$ must satisfy
the following two conditions:\begin{eqnarray} &&
\frac{\gamma_k(\bar{\mv{W}},\bar{\mv{p}})}{\bar{\gamma}_k}=C(\bar{\mv{W}}),
\ \forall k, \label{equ:equal sinr balancing level}\\ &&
\bar{p}_i=P_i^{{\rm max}},\label{eqn:equal power}\end{eqnarray}where
$C(\bar{\mv{W}})$ is the maximum SINR balancing value for all users
given $\bar{\mv{W}}$.

We justify the above conditions as follows. (\ref{equ:equal sinr
balancing level}) can be shown by contradiction. Supposing that the
SINR balancing values are not the same for all the users, then we
select the user with the highest SINR balancing value and decrease
its transmit power by a small amount such that its new SINR
balancing value is still above
$\min\limits_k\frac{\gamma_k}{\bar{\gamma}_k}$. Since the other
users' SINR balancing values will increase, the minimum SINR
balancing value among all the users will increase accordingly. Thus,
whenever the SINR balancing values are not the same for all users,
we can proceed as above to improve the optimal value. Hence,
(\ref{equ:equal sinr balancing level}) must hold. Similarly to show
(\ref{eqn:equal power}) by contradiction, suppose
$\bar{p}_i<P_i^{{\rm max}}$. With $\alpha=\frac{P_i^{{\rm
max}}}{\bar{p}_i}>1$, we can multiply the transmit power values of
each user by $\alpha$, and the SINRs of all users will be increased
accordingly. Hence, (\ref{eqn:equal power}) must hold.

We can express (\ref{equ:equal sinr balancing level}) for all $k$'s
in the following matrix form:\begin{align}\label{equ:matrix
form1}\bar{\mv{p}}\frac{1}{C(\bar{\mv{W}})}=\mv{D} \mv
\Psi(\bar{\mv{W}})\bar{\mv{p}}+\mv{D}\mv \sigma,\end{align}where
$\mv{D}={\rm Diag}\{\frac{\bar{\gamma}_1}{\|\bar{\mv{w}}_1^H\mv{
h}_{1,1}\|^2},\cdots,\frac{\bar{\gamma}_K}{\|\bar{\mv{w}}_K^H\mv{
h}_{K,K}\|^2}\}$, $\mv \sigma=[\sigma_1^2\|\bar{\mv{
w}}_1\|^2,\cdots,\sigma_K^2\|\bar{\mv{w}}_K\|^2]^T$, and the $K
\times K$ non-negative matrix $\mv \Psi(\bar{\mv{W}})$ is a function
of $\bar{\mv{W}}$ defined
as\begin{align}[\Psi(\bar{\mv{W})}]_{k,j}=\left\{\begin{array}{ll}\|\bar{\mv{
w}}_k^H\mv{h}_{k,j}\|^2, & k\neq j\\ 0, &
k=j.\end{array}\right.\end{align}

By multiplying both sides of (\ref{equ:matrix form1}) by
$\mv{e}_i^T$, we obtain
\begin{align}\label{equ:power constraint}\mv{e}_i^T\bar{\mv{p}}\frac{1}{C(\bar{\mv{W}})}=\frac{P_i^{{\rm max}}}{C(\bar{\mv{W}})}=\mv{e}_i^T\mv{D} \mv \Psi(\bar{\mv{W}})\bar{\mv{p}}+\mv{e}_i^T\mv{D}\mv \sigma.\end{align}
Therefore, by combining (\ref{equ:matrix form1}) and (\ref{equ:power
constraint}), it follows that
\begin{align}\label{equ:matrix form2}\frac{1}{C(\bar{\mv{W}})}\bar{\mv{p}}_{{\rm ext}}=
\mv{A}_i(\bar{\mv{W}})\bar{\mv{p}}_{{\rm ext}},\end{align}where $\bar{\mv{p}}_{{\rm ext}}=\left(\begin{array}{c}\bar{\mv{p}} \\
1\end{array}\right)$ and
\begin{align}\label{eqn:1}\mv{A}_i(\bar{\mv{
W}})=\left(\begin{array}{cc}\mv{D}\mv \Psi(\bar{\mv{W}}) & \mv{
D}\mv \sigma\\ \frac{1}{P_i^{{\rm max}}}\mv{e}_i^T\mv{D}\mv
\Psi(\bar{\mv{W}}) & \frac{1}{P_i^{{\rm max}}}\mv{e}_i^T\mv{D}\mv
\sigma\end{array}\right).\end{align}

Next, we show one important property for (\ref{equ:matrix form2}) in
the following lemma.\begin{lemma}\label{lem:unique solution} Given
any fixed $\bar{\mv{W}}$, there exists a unique solution
$(\bar{\mv{p}},C(\bar{\mv{W}}))$ to the equation in (\ref{equ:matrix
form2}).\end{lemma}

\begin{proof}
Please refer to Appendix \ref{appendix:proof of lemma}.
\end{proof}

According to Perron-Frobenius theory \cite{Horn85}, for any
nonnegative matrix, there is at least one positive eigenvalue and
the spectral radius of the matrix is equal to the largest positive
eigenvalue. Furthermore, according to Lemma \ref{lem:unique
solution}, there is only one strictly positive eigenvalue to matrix
$\mv{A}_i(\bar{\mv{W}})$. Accordingly, it follows from
(\ref{equ:matrix form2}) that given $\bar{\mv{W}}$, the inverse of
the optimal SINR balancing value $1/C(\bar{\mv{W}})$ is the spectral
radius of $\mv{A}_i(\bar{\mv{W}})$. Consequently, the maximum SINR
balancing solution to Problem (\ref{equ:sub-problem}) is obtained as
\begin{align}\label{eqn:optimal balance}C^\ast=\frac{1}{\min\limits_{\mv{W}}\rho(\mv{A}_i(\mv{
W}))}.\end{align}

Next, by defining a cost function as
\begin{align}\Upsilon(\mv{W},\mv{p}_{{\rm
ext}})=\max\limits_{\mv{x}>0}\frac{\mv{x}^T\mv{A}_i(\mv{W})\mv{p}_{{\rm
ext}}}{\mv{x}^T\mv{p}_{{\rm ext}}},\end{align}then the min-max
characterization of the spectral radius of $\mv{A}_i(\mv{W})$ can be
expressed as \cite{Schubert04}, \cite{Codreanu07}
\begin{align}\label{eqn:maxmin}\rho(\mv{A}_i(\mv{ W}))=\min\limits_{\mv{p}_{{\rm
ext}}}\Upsilon(\mv{W},\mv{p}_{{\rm ext}}).\end{align}Taking
(\ref{eqn:maxmin}) into (\ref{eqn:optimal balance}), it follows that
\begin{align}\label{eqn:maxmin problem}\frac{1}{C^\ast}=\min\limits_{\mv{W}}\min\limits_{\mv{p}_{{\rm
ext}}}\Upsilon(\mv{W},\mv{p}_{{\rm ext}}).\end{align}

Similar to \cite{Schubert04}, we can solve Problem (\ref{eqn:maxmin
problem}) via the alternating optimization shown as follows. First,
given $\bar{\mv{W}}$, we find the optimal power allocation for
$\mv{p}_{{\rm ext}}$. Let $\bar{\mv{p}}_{{\rm ext}}$ denote the
dominant eigenvector corresponding to the spectral radius of
$\mv{A}_i(\bar{\mv{W}})$. It then follows that
\begin{align}\frac{\mv{x}^T\mv{A}_i(\bar{\mv{W}})\bar{\mv{p}}_{{\rm
ext}}}{\mv{x}^T\bar{\mv{p}}_{{\rm
ext}}}=\rho(\mv{A}_i(\bar{\mv{W}}))=\min\limits_{\mv{p}_{{\rm
ext}}}\Upsilon(\bar{\mv{W}},\mv{p}_{{\rm ext}}).\end{align}Thus,
$\bar{\mv{p}}_{{\rm ext}}$ is the optimal power allocation given
$\bar{\mv{W}}$.

Furthermore, we know that given any power allocation $\mv{p}_{{\rm
ext}}$, the optimal receive beamformer in $\mv{W}$ to maximize the
SINR is minimum-mean-squared-error (MMSE) based for each of the
users. Therefore, we propose an iterative algorithm in Table
\ref{table3}, denoted as Algorithm \ref{table3}, to solve Problem
(\ref{equ:sub-problem}).

\begin{table}[htp]
\begin{center}
\caption{\textbf{Algorithm \ref{table3}}: Algorithm for Solving
Problem (\ref{equ:sub-problem})} \vspace{0.2cm}
 \hrule
\vspace{0.2cm}
\begin{itemize}
\item[a)] Initialize: $n=0$, $\mv{p}^{(0)}=[0,\cdots,0]^T$ and $\rho^{(0)}=\infty$;
\item[b)] {\bf repeat}
\begin{itemize}
\item[1)] $n=n+1$;
\item[2)] Update $\mv{W}^{(n)}$ by $\mv{w}_k^{(n)}=(\sum\limits_{j\neq
k}p_j^{(n-1)}{\bf h}_{k,j}{\bf h}_{k,j}^H+\sigma_k^2{\bf
I})^{-1}{\bf h}_{k,k}, ~~ \forall k$;
\item[3)] Update $\mv{p}^{(n)}_{{\rm ext}}$ as the dominant eigenvector of the matrix $\mv{A}_i(\mv{W}^{(n)})$;
\item[4)] $\rho^{(n)}=\rho(\mv{A}_i(\mv{W}^{(n)}))$ and $C^{(n)}=\frac{1}{\rho^{(n)}}$;
\end{itemize}
\item[c)] {\bf until} $\rho^{(n-1)}-\rho^{(n)}<\epsilon$.
\end{itemize}
\vspace{0.2cm} \hrule \label{table3}
\end{center}
\end{table}

The convergence of Algprithm \ref{table3} can be shown in the
following way. Since given any power allocation $\mv{p}_{{\rm
ext}}^{(n)}$ for the $n$th iteration, $\mv{W}^{(n+1)}$ minimizes
$\Upsilon(\mv{W},\mv{p}_{{\rm ext}}^{(n)})$, i.e.,
\begin{align}\Upsilon(\mv{W}^{(n+1)},\mv{p}_{{\rm ext}}^{(n)})\leq \Upsilon(\mv{W}^{(n)},\mv{p}_{{\rm
ext}}^{(n)})=\rho^{(n)}.\end{align}Moreover, given $\mv{W}^{(n+1)}$,
$\mv{p}_{{\rm ext}}^{(n+1)}$ minimizes
$\Upsilon(\mv{W}^{(n+1)},\mv{p}_{{\rm ext}})$ as
\begin{align}\rho^{(n+1)}=\Upsilon(\mv{W}^{(n+1)},\mv{p}_{{\rm
ext}}^{(n+1)})\leq \Upsilon(\mv{W}^{(n+1)},\mv{p}_{{\rm
ext}}^{(n)}).\end{align}Hence, we can guarantee $\rho^{(n+1)}\leq
\rho^{(n)}$ after each iteration. Since $\rho$ is lower-bounded by
$0$, Algorithm \ref{table3} thus converges.

Finally, the convergence of Algorithm \ref{table3} to the global
optimality of Problem (\ref{equ:sub-problem}) can be proven
similarly as Section IV.A in \cite{Schubert04}, and the proof is
thus omitted for brevity. After convergence, $C^{(n)}\bar{\mv
\gamma}$ is the maximum achievable SINR vector and $\mv{p}^{(n)}$,
$\mv{W}^{(n)}$ are the optimal power and receive beamforming vectors
to achieve this SINR vector, respectively.

\subsubsection{Solution to Problem (\ref{equ:SINR balancing problem for SIMO-IC with
individual power constraint})}

\ \ \

Next, we show that Problem (\ref{equ:SINR balancing problem for
SIMO-IC with individual power constraint}) can be efficiently solved
via solving Problem (\ref{equ:sub-problem}) for all $i$'s. Let
$\mv{W}^\ast$ and $\mv{p}^\ast$ denote the optimal beamforming
vectors and power allocation for Problem (\ref{equ:SINR balancing
problem for SIMO-IC with individual power constraint}),
respectively. Let $\mv{W}^\ast_i$ and $\mv{p}^\ast_i$ denote the
optimal beamforming vectors and power allocation for the $i$th
sub-problem in (\ref{equ:sub-problem}), respectively. Next, we
provide a theorem to reveal the relationship between the optimal
solutions to Problems (\ref{equ:SINR balancing problem for SIMO-IC
with individual power constraint}) and (\ref{equ:sub-problem}).

\begin{theorem}\label{the:jonit beamforming and power allocation}
For all sub-problems in (\ref{equ:sub-problem}) with $i=1,\cdots,K$,
there exists one and only one sub-problem for which the optimal
power solution satisfies all users' individual power constraints of
Problem (\ref{equ:SINR balancing problem for SIMO-IC with individual
power constraint}). Furthermore, let $i^\ast$ denote the index of
the corresponding sub-problem in (\ref{equ:sub-problem}), then it
holds that $\mv{W}^\ast=\mv{ W}^\ast_{i^\ast}$, and
$\mv{p}^\ast=\mv{p}^\ast_{i^\ast}$.
\end{theorem}

\begin{proof}
Please refer to Section IV.B in \cite{Zhanglan08}.
\end{proof}

Theorem \ref{the:jonit beamforming and power allocation} reveals
that Problem (\ref{equ:SINR balancing problem for SIMO-IC with
individual power constraint}) can be solved as follows. First, we
apply Algorithm \ref{table3} to solve Problem
(\ref{equ:sub-problem}) in the order of $i=1,\cdots,K$. If the
optimal power solution to any of these problems satisfies all users'
individual power constraints, the algorithm terminates, and the
obtained optimal power and beamforming solutions to Problem
(\ref{equ:sub-problem}) are also those to Problem (\ref{equ:SINR
balancing problem for SIMO-IC with individual power constraint}).
The above algorithm, denoted by Algorithm \ref{table4}, is
summarized in Table \ref{table4}.

\begin{table}[htp]
\begin{center}
\caption{\textbf{Algorithm \ref{table4}}: Algorithm for Solving
Problem (\ref{equ:SINR balancing problem for SIMO-IC with individual
power constraint})} \vspace{0.2cm}
 \hrule
\vspace{0.2cm}
\begin{itemize}
\item[a)] Initialize: $i=0$;
\item[b)] {\bf repeat}
\begin{itemize}
\item[1)] $i=i+1$;
\item[2)] Solve the $i$th sub-problem in (\ref{equ:sub-problem}) by Algorithm \ref{table3}, and find the optimal beamforming solution $\mv{W}^\ast_i$ and power solution $\mv{p}^\ast_i$;
\item[3)] Check whether $\mv{p}^\ast_i$ satisfies all power
constraints of Problem (\ref{equ:SINR balancing problem for SIMO-IC
with individual power constraint}). If so, exit the algorithm and
set $\mv{W}^\ast_i$ and $\mv{p}^\ast_i$ as the optimal solution to
Problem (\ref{equ:SINR balancing problem for SIMO-IC with individual
power constraint}); otherwise, continue the algorithm;
\end{itemize}
\item[c)] {\bf until} $i=K$.
\end{itemize}
\vspace{0.2cm} \hrule \label{table4}
\end{center}
\end{table}

\subsection{The MISO-IC Case}\label{sec:MISO-IC}

In this subsection, we show how to solve the feasibility problem
(P3.3) for MISO-IC under the equivalent SINR constraints given by
(\ref{equ:sinr feasibility}). It was shown in \cite{Rui10} that this
problem can be transformed into a second-order cone programming
(SOCP) problem, which is briefly described as follows for the sake
of completeness. The SINR constraints in Problem (P3.3) can be
rewritten as
\begin{align}(1+\frac{1}{\bar{\gamma}_k})\|\mv{h}_{k,k}^H \mv{
v}_k\|^2\geq
\sum\limits_{j=1}^K\|\mv{h}_{k,j}^H\mv{v}_j\|^2+\sigma_k^2, ~~
\forall k. \end{align}

Without loss of generality, we can assume that
$\mv{h}_{k,k}^H\mv{w}_k$ is a positive real number, $\forall k$.
Thus we can reformulate the above SINR constraints
as\begin{align}\label{equ:transform1}\sqrt{1+\frac{1}{\bar{\gamma}_k}}\mv{h}_{k,k}^H
\mv{ v}_k\geq
\sqrt{\sum\limits_{j=1}^K\|\mv{h}_{k,j}^H\mv{v}_j\|^2+\sigma_k^2},
~~ \forall k.\end{align}

Denote $\mv{x}=[\mv{v}_1^T,\cdots,\mv{v}_K^T,0]^T$ of dimension
$(K^2+1)\times 1$, $\mv{n}_k=[0,\cdots,0,\sigma_k]^T$ of dimension
$(K+1)\times 1$, and $\mv{E}_k={\rm
Diag}(\mv{h}_{k,1}^H,\cdots,\mv{h}_{k,K}^H,0)$ of dimension
$(K+1)\times(K^2+1)$, $\forall k$. We further define $\mv{L}_k$ as
\begin{align}\mv{L}_k=\bigg[\underbrace{\mv{0}^{K\times K},\cdots,\mv{0}^{K\times
K}}_{k-1},\mv{I}^{K\times K},\underbrace{\mv{0}^{K\times
K},\cdots,\mv{0}^{K\times K}}_{K-k},\mv{0}^{K\times
1}\bigg],\end{align}where $\mv{0}^{K\times K}$ and $\mv{0}^{K\times
1}$ denote the $K\times K$ all-zero matrix and $K\times 1$ all-zero
vector, respectively, and $\mv{I}^{K\times K}$ denotes the $K\times
K$ identity matrix. Thus, (\ref{equ:transform1}) can be reformulated
as\begin{align}\label{equ:transform2}\left\|\mv{E}_k\mv{x}+\mv{n}_k\right\|
\leq \sqrt{1+\frac{1}{\bar{\gamma}_k}}\mv{h}_{k,k}^H\mv{L}_k\mv{x},
~~ \forall k.\end{align}

Moreover, we can reformulate the power constraints
as\begin{align}\label{equ:transform3}\|\mv{L}_k\mv{x}\| \leq
\sqrt{P_k^{{\rm max}}}, ~~ \forall k. \end{align}

Using (\ref{equ:transform2}) and (\ref{equ:transform3}), Problem
(P3.3) can be transformed into a SOCP feasibility problem over
$\mv{x}$ and efficiently solvable by existing software \cite{Grant}.

\begin{remark}It is worth comparing our proposed algorithm with that in \cite{Jorswieck10} for solving the WSR maximization problem (P1.3) for MISO-IC. The algorithm in
\cite{Jorswieck10} is based on a prior result in \cite{Jorswieck08}
that for the special case of two-user MISO-IC, any point on the
Pareto boundary of the rate region can be achieved by transmit
beamforming vectors that are obtained by linearly combining the ZF
and MRT beamformers. In \cite{Jorswieck10}, the outer polyblock
approximation algorithm was applied to find the optimal beamformer
combining coefficients. However, since this result does not hold for
MISO-IC with more than two users, the algorithm in
\cite{Jorswieck10} cannot be extended to the general $K$-user
MISO-IC with $K>2$. In contrast, our proposed algorithm can be
applied to MISO-IC with arbitrary number of users.
\end{remark}

\section{Numerical Results}\label{sec:numerical results}

In this section, we provide numerical results to validate the
proposed algorithms in this paper. We assume that $\mu_k=1$,
$\forall k$, i.e., the sum-rate maximization problem is considered.
We also assume that $P_k^{{\rm max}}=3$, $\forall k$. For SIMO-IC
and MISO-IC, we further assume that $M_k=2$ and $N_k=2$,
respectively, $\forall k$. The numerical results with related
discussions are presented in the following subsections.

\subsection{Convergence Performance}\label{subsec:convergence
performance} Firstly, we study the convergence performance of
Algorithm \ref{table1} for SISO-IC. We assume that there are $4$
users, i.e., $K=4$, and there is a minimum rate constraint for each
user with $R_k^{{\rm min}}=0.5$, $\forall k$. We set the parameters
to control the accuracy of Algorithm \ref{table1} as $\epsilon=0.01$
and $\eta=0.5$. We consider the following
matrix:\begin{align}\label{equ:channel
gain}\mv{H}=\left[\begin{array}{cccc}0.4310 & 0.0022 & 0.0105 & 0.0042\\
0.0200 & 0.4102 & 0.0180 & 0.0035\\0.0210 & 0.0200 & 0.5162 &
0.0112\\ 0.0210 & 0.0021 & 0.0063 &
0.3634\end{array}\right],\end{align}with each element denoting the
power of the corresponding channel gain, i.e.,
$H_{k,j}=\|h_{k,j}\|^2$.

\begin{figure}
\begin{center}
\scalebox{0.9}{\includegraphics*[104pt,248pt][494pt,549pt]{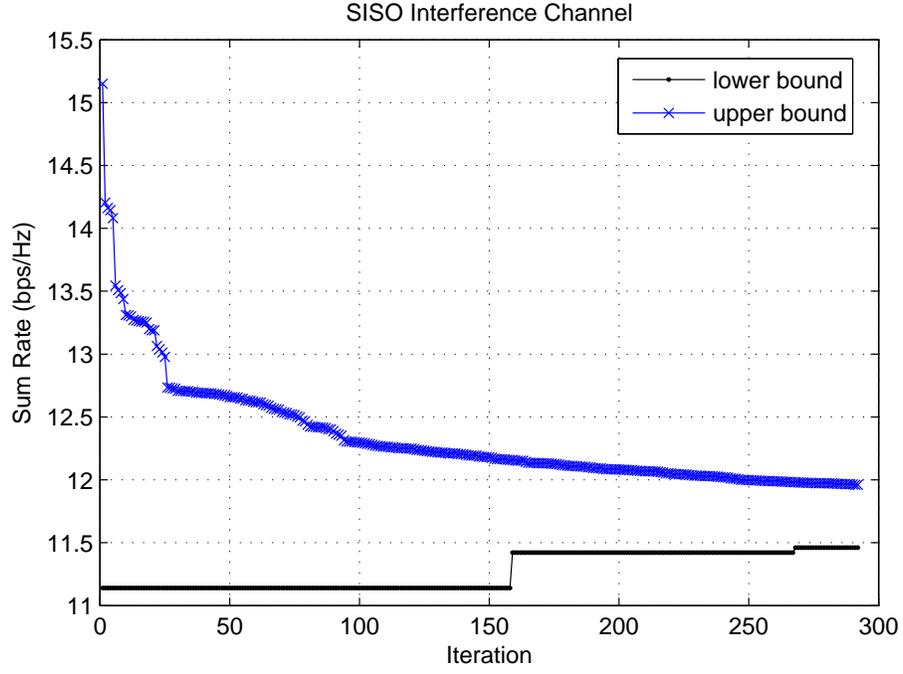}}
\end{center}
\caption{Convergence performance of Algorithm \ref{table1} for
SISO-IC with weak interference channel gains.}\label{fig3}
\end{figure}

\begin{figure}
\begin{center}
\scalebox{0.9}{\includegraphics*[110pt,249pt][498pt,548pt]{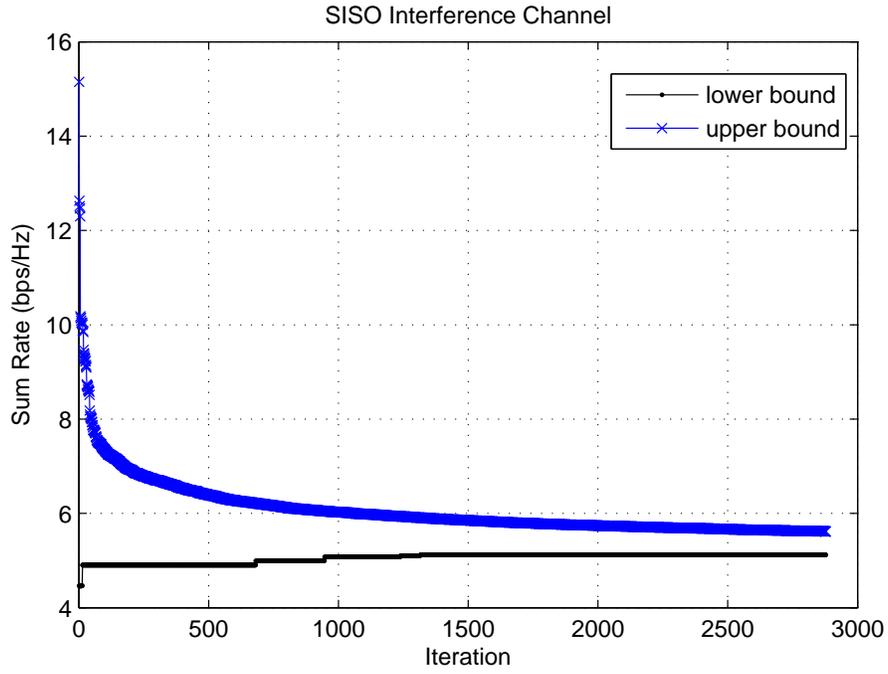}}
\end{center}
\caption{Convergence performance of Algorithm \ref{table1} for
SISO-IC with strong interference channel gains.}\label{fig4}
\end{figure}

Fig. \ref{fig3} shows the convergence of Algorithm \ref{table1}
under the above channel setup. It is observed that this algorithm
takes about $300$ iterations to converge. The converged sum-rate is
$11.4605$ with users' individual rates given by $[3.1982,2.6297,
2.8441,2.7884]$. To verify that the global sum-rate maximum is
achieved, we compare the obtained maximum sum-rate with that by an
exhaustive search, which is equal to $11.5349$. Thus, Algorithm
\ref{table1} does achieve the global optimality of sum-rate
maximization within a guaranteed error $11.5349-11.4605=0.0744$,
which is smaller than the set threshold $\eta=0.5$.

Next, we consider a SISO-IC with stronger cross-user interference
channel gains than those given in (\ref{equ:channel gain}) by
keeping all diagonal elements of $\mv{H}$ unchanged, but scaling all
off-diagonal elements by $10$ times. As shown in Fig. \ref{fig4},
for this new channel setup, Algorithm \ref{table1} takes about
$2900$ iterations to converge. The converged sum-rate in this case
is $5.1184$ with users' individual rates given by $[0.5408, 1.9119,
0.5060, 2.1597]$, while that obtained by the exhaustive search is
$5.1392$. Thus, as compared to the previous case with weaker
interference channel gains, the global optimality for sum-rate
maximization is achieved in this case with a much slower
convergence. The reason is as follows. With stronger interference
channel gains, the optimal power allocation for sum-rate
maximization is more likely to render some users transmit at their
minimum required rates (e.g., user $1$ and user $3$ in this
example). Hence, the corresponding optimal rate values will lie in
the strip defined by $\{\mv{r}^\ast|R_k^{{\rm min}}\leq R_k^\ast
\leq R_k^{{\rm min}}+\epsilon \}$ for some $k$'s. Since in Algorithm
\ref{table1} each new polyblock is generated from the previous one
by cutting off some unfit portions, the cuts become shallower and
shallower as $\tilde{\mv{z}}^{(n)}$ approaches the above strip. This
can be observed from Fig. \ref{fig4} that after the $1300$th
iteration, the best intersection point $\tilde{\mv{r}}^{(n)}$ has
never changed. However, to make
$U(\tilde{\mv{z}}^{(n)})-U(\tilde{\mv{r}}^{(n)})\leq \eta$ hold,
another $1700$ iterations are taken just to reduce the value of
$U(\tilde{\mv{z}}^{(n)})$. Since this reduction becomes very
inefficient near the strip, the algorithm converges much more slowly
to the desired accuracy with the increasing of interference channel
gains. From this observation, we infer that the values of $\epsilon$
and $\eta$ need to be properly set to balance between the accuracy
and convergence speed of our proposed algorithm.

Next, we give another example to illustrate the important role of
parameter $\epsilon$ in balancing between the accuracy and
convergence speed of our proposed algorithm. We assume that $K=3$,
and there are no minimum rate requirements for the users. We
consider the following channel
matrix:\begin{align}\label{equ:channel
gain1}\mv{H}=\left[\begin{array}{ccc}0.4310 & 0.0187 & 0.0893 \\
0.1700 & 0.4102 & 0.1530 \\0.1785 & 0.1700 & 0.5162
\end{array}\right],\end{align}with $H_{k,j}=\|h_{k,j}\|^2$. By an
exhaustive search, the optimal sum-rate is obtained as $4.8079$ with
users' individual rates given by $[3.2146,1.5933,0]$.

\begin{table}
\caption{Selection of $\epsilon$ on the Performance of the Proposed
Algorithm} \label{table5}
\begin{center}
\begin{tabular}{c|c|c}
\hline Value of $\epsilon$ & Number of iterations & Converged WSR \\
\hline\hline
$0.05$ & $8183$ & $4.7625$ \\
$0.10$ & $3498$ & $4.7438$ \\
$0.15$ & $2212$ & $4.7275$ \\
$0.20$ &  $1642$ & $4.6942$ \\
$0.25$ & $1396$ & $4.6825$ \\
$0.30$ & $1148$ & $4.6620$ \\
$0.35$ & $1029$ & $4.6350$ \\
$0.40$ & $866$ & $4.6165$ \\
$0.45$ & $651$ & $4.5880$ \\

 \hline
\end{tabular}
\end{center}
\end{table}

Table \ref{table5} shows the convergence speed and the converged
sum-rate of our proposed algorithm for different values of
$\epsilon$ with $\eta=0.2$. We observe that as $\epsilon$ increases,
the algorithm convergence speed improves rapidly, but the converged
sum-rate decreases. When $\epsilon=0.45$, the difference between the
optimal sum-rate and converged sum-rate is $4.8079-4.5880=0.2199$,
which is even larger than $\eta=0.2$. This is because that as we
show in Section \ref{subsec:algorithm}, with non-zero $\epsilon$, we
are in fact solving Problem P2-A instead of the original problem P2.
Consequently, the proposed algorithm can only guarantee that the
difference between the maximum sum-rate of Problem P2-A and the
converged sum-rate is less than $\eta$, but not necessarily for
Problem P2. Thus, if the value of $\epsilon$ is selected to be too
large such that all the $\eta$-optimal solutions lie in the excluded
strips, the difference of the converged sum-rate and the maximum
sum-rate of Problem (P2) will be larger than $\eta$. Therefore, the
value of $\epsilon$ should be carefully selected based on the value
of $\eta$. In this numerical example, we can select $\epsilon=0.40$
such that the $\eta$-optimal solution is still guaranteed and also
the converged speed is reasonably fast.

\subsection{Providing Performance Benchmark for Other Heuristic Algorithms}

\begin{figure}
\begin{center}
\scalebox{0.9}{\includegraphics*[100pt,274pt][489pt,566pt]{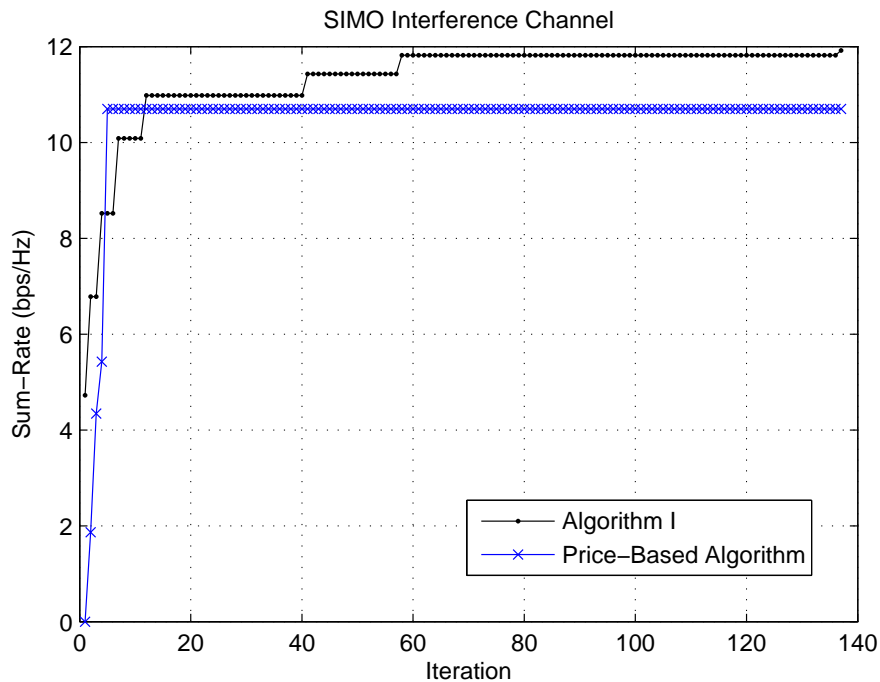}}
\end{center}
\caption{Performance comparison for Algorithm \ref{table1} versus
the price-based algorithm in SIMO-IC.}\label{fig5}
\end{figure}

\begin{figure}
\begin{center}
\scalebox{0.9}{\includegraphics*[98pt,275pt][483pt,575pt]{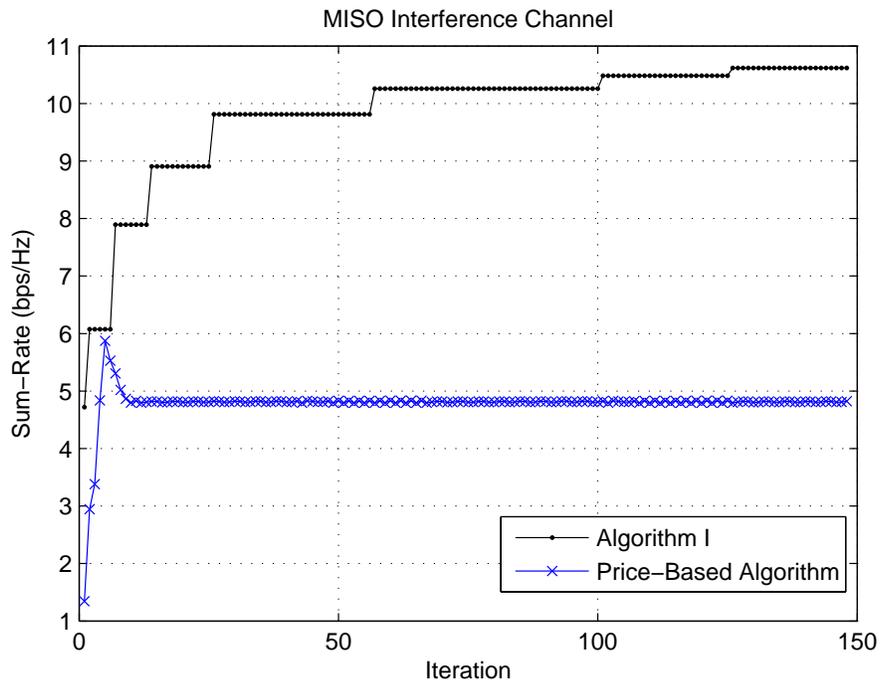}}
\end{center}
\caption{Performance comparison for Algorithm \ref{table1} versus
the price-based algorithm in MISO-IC.}\label{fig6}
\end{figure}

A key application of our proposed algorithm is to provide
performance benchmarks for other heuristic algorithms for achieving
the maximum WSR in the GIC, especially in cases of MISO-IC and
SIMO-IC where the globally optimal solution by exhaustive search is
hardly possible. In the following, we provide an example to show how
to utilize our proposed algorithm to evaluate the performance of
other suboptimal algorithms for WSR maximization in MISO-IC and
SIMO-IC.

We consider the ``price-based'' suboptimal algorithm, e.g., the ADP
algorithm, which was proposed in \cite{Huang06} as an efficient
distributed algorithm for WSR maximization in SISO-IC. Since to our
best knowledge, extensions of ADP to the multi-antenna GIC are not
yet available in the literature, we provide the details for such
extensions for SIMO-IC and MISO-IC in Appendix \ref{ADP}.

Figs. \ref{fig5} and \ref{fig6} show the achievable sum-rates by the
price-based algorithm versus Algorithm \ref{table1} for $4$-user
SIMO-IC and MISO-IC, respectively, without the minimum rate
constraints. Each element in all channel vectors involved is
randomly generated by the CSCG distribution with zero mean and unit
variance. We set the parameters to control the accuracy of Algorithm
\ref{table1} as $\epsilon=0.01$ and $\eta=0.5$. In Fig. \ref{fig5},
the price-based algorithm converges to the sum-rate of $10.6989$ in
SIMO-IC, while the maximum sum-rate achieved by Algorithm
\ref{table1} is $11.9182$. In Fig. \ref{fig6}, the price-based
algorithm converges to the sum-rate of $4.8216$ (although it has
reached almost $6$ before convergence) in MISO-IC, while Algorithm
\ref{table1} achieves the maximum sum-rate of $10.6193$. Based on
these results as well as other numerical examples (not shown in this
paper due to the space limitation), we infer that in general the
price-based algorithm for SIMO-IC performs better than MISO-IC, as
both compared with our proposed  algorithm that achieves the global
sum-rate maximum. Moreover, the price-based algorithm for MISO-IC
does not converge under certain channel setups, while even when the
algorithm converges, the resulted sum-rate can be far from the
global maximum. In contrast, for SIMO-IC, the price-based algorithm
usually achieves the sum-rate very close to the global maximum, and
even converges to it under certain channel setups.

\section{Concluding Remarks}\label{sec:conclusion}

In this paper, we propose a new framework to achieve the global
optimality of WSR maximization problems in SISO-IC, SIMO-IC, and
MISO-IC, with the interference treated as Gaussian noise. Although
the studied problems are non-convex with respect to the power
allocation and/or beamforming vectors, we show that they belong to
the monotonic optimization over a normal set by reformulating them
as maximizing the WSR in the achievable rate regions directly.
Therefore, the outer polyblock approximation algorithm can be
applied to achieve the global WSR maximum. Furthermore, by utilizing
the approach of rate profile, at each iteration of the proposed
algorithm, the updated intersection point on the Pareto boundary of
the achievable rate region is efficiently obtained via solving a
sequence of SINR feasibility problems. It is worth noting that
although the developed framework in this paper is aimed to solve the
WSR maximization problem for the GIC, it can be similarly applied to
other multiuser communication systems with non-convex rate regions
provided that the problem of characterizing the intersection Pareto
boundary point with an arbitrary rate-profile vector can be
efficiently solved.

It is worth pointing out that based on our numerical experiments,
the proposed algorithm in this paper is found to converge very
slowly when the number of users becomes large (e.g., $K \geq 6$),
and thus may not be suitable for real-time implementation.
Nevertheless, the proposed algorithm can be applied to provide
performance benchmarks for other real-time algorithms that usually
guarantee only suboptimal solutions. It is our hope that this paper
will motivate future work to improve the convergence speed of the
proposed algorithm and thus make it more applicable in practical
systems, even with large number of users. For example, in
\cite{Audet05}, the original point for the algorithm is shifted from
the origin to a point in the negative plane, which is shown to speed
up the convergence to some extent.

After the submission of this manuscript, we become aware of one
interesting related work \cite{Otterstan} that is worth mentioning.
In \cite{Otterstan}, a similar framework is proposed to optimize the
system performance for multi-cell downlink MISO beamforming (similar
to MISO-IC in nature), e.g., sum-rate performance and proportional
fairness, by making use of monotonic optimization and rate profile
techniques. One difference between \cite{Otterstan} and our work is
that for the monotonic optimization part, a so-called
``branch-reduce-and-bound'' algorithm is used in \cite{Otterstan} as
compared to the outer polyblock approximation algorithm in our
paper.

\begin{appendix}

\subsection{Proof of Lemma \ref{lem:unique solution}}\label{appendix:proof of lemma}

Note that under the sum power constraint, a similar result to this
lemma has been shown in \cite{Xu98}. However, the proof in
\cite{Xu98} is not directly applicable in our case since in
(\ref{equ:matrix form2}), there is an individual power constraint
rather than the sum power constraint. Thus, we need to provide a new
proof for this lemma shown as follows.

Suppose that there are two solutions to (\ref{equ:matrix form2}),
denoted by $(\bar{\mv{p}},C(\bar{\mv{W}}))$ and
$(\bar{\mv{p}}',C'(\bar{\mv{W}}))$. Define a sequence of
$\theta_k$'s as $\theta_k=\frac{\bar{p}_k'}{\bar{p}_k}$, $\forall
k$. We can re-arrange $\theta_k$'s in a decreasing order
by\begin{align}\label{eqn:sequence}\theta_{t1}\geq \theta_{t2} \geq
\cdots \geq \theta_{tK}.\end{align}Since according to
(\ref{eqn:equal power}) we have $\bar{p}_i=\bar{p}_i'=P_i^{{\rm
max}}$, it follows that $\theta_i=1$ must hold. Hence,
$\theta_{t1}\geq \theta_i=1$. Moreover, in (\ref{eqn:sequence}), at
least one strict inequality must hold because otherwise
$\theta_k=1$, $\forall k$, which then implies that only one unique
solution to (\ref{equ:matrix form2}) exists.

Next, we derive the SINR balancing value of user $t1$ as
follows:\begin{align}C_{t1}'(\bar{\mv{W}})&=\frac{\bar{p}_{t1}'\|\bar{\mv{w}}_{t1}^H\mv{
h}_{t1,t1}\|^2}{\bar{\mv{w}}_{t1}^H(\sum\limits_{j\neq
t1}\bar{p}_j'\mv{h}_{t1,j}\mv{
h}_{t1,j}^H+\sigma_{t1}^2\mv{I})\bar{\mv{w}}_{t1}\bar{\gamma}_{t1}}
\nonumber \\ &=\frac{\bar{p}_{t1}\|\bar{\mv{w}}_{t1}^H\mv{
h}_{t1,t1}\|^2}{\bar{\mv{w}}_{t1}^H(\sum\limits_{j\neq
t1}\bar{p}_j\mv{h}_{t1,j}\mv{
h}_{t1,j}^H\frac{\theta_j}{\theta_{t1}}+\sigma_{t1}^2\mv{I}\frac{1}{\theta_{t1}})\bar{\mv{w}}_{t1}\bar{\gamma}_{t1}}
\nonumber \\ &>\frac{\bar{p}_{t1}\|\bar{\mv{w}}_{t1}^H\mv{
h}_{t1,t1}\|^2}{\bar{\mv{w}}_{t1}^H(\sum\limits_{j\neq
t1}\bar{p}_j\mv{h}_{t1,j}\mv{
h}_{t1,j}^H+\sigma_{t1}^2\mv{I})\bar{\mv{w}}_{t1}\bar{\gamma}_{t1}}
\nonumber \\ & =C_{t1}(\bar{\mv{W}}).\end{align}Based on
(\ref{equ:equal sinr balancing level}), we have
\begin{align}\label{eqn:greater}C'(\bar{\mv{W}})=C'_{t1}(\bar{\mv{W}})>C_{t1}(\bar{\mv{W}})=C(\bar{\mv{W}}).\end{align}Similarly, we can show that
$C'_{tK}(\bar{\mv{W}})<C_{tK}(\bar{\mv{W}})$, which
yields\begin{align}\label{eqn:less}C'(\bar{\mv{W}})=C'_{tK}(\bar{\mv{W}})<C_{tK}(\bar{\mv{W}})=C(\bar{\mv{W}}).\end{align}Since
(\ref{eqn:greater}) and (\ref{eqn:less}) contradict to each other,
there must be one unique solution to (\ref{equ:matrix form2}). Lemma
\ref{lem:unique solution} is thus proven.

\subsection{Price-Based Algorithm for SIMO-IC and
MISO-IC}\label{ADP} In this part, we provide the details of the
suboptimal price-based algorithms for Problems (P1.2) in SIMO-IC and
(P1.3) in MISO-IC, which can be viewed as extensions of the ADP
algorithm proposed in \cite{Huang06} for SISO-IC. In ADP, each user
announces a price that reflects its sensitivity to the interference
from all other users, and then updates its transmit power by
maximizing its own utility offset by the sum interference price
received from all the other users. It was shown in \cite{Huang06}
that ADP can converge to the solution that has the same
Karush-Kuhn-Tucker (KKT) conditions as that of the WSR maximization
problem, and is thus guaranteed to achieve at least a locally
optimal solution. In the following, we extend the ADP algorithm in
\cite{Huang06} to SIMO-IC and MISO-IC, but without the proof of
convergence.

\subsubsection{Price-Based Algorithm for
SIMO-IC}\label{appendix:interference pricing SIMO-IC}

\ \ \

In this part, we extend the ADP or price-based algorithm to SIMO-IC.
First, without loss of generality, we substitute the optimal
MMSE-based receive beamforming vectors for $\mv{w}_k$'s into the
SINR expression given in (\ref{equ:SINR in SIMO-IC}). Then, given
any transmit power vector $\mv{p}$, the achievable rate for user $k$
can be expressed as
\begin{align}R_k(\mv{p})=\log_2(1+\gamma_k^{{\rm
SIMO-IC}})=\log_2\bigg(1+p_k\mv{h}_{k,k}^H(\sum\limits_{j\neq
k}p_j\mv{h}_{k,j}\mv{h}_{k,j}^H+\sigma_k^2\mv{I})^{-1}\mv{h}_{k,k}\bigg).\label{eq:MMSE
SINR}\end{align}Thus in Problem (P1.2), we only need to find the
optimal transmit power solution, without the need to consider the
receive beamforming optimization.

Next, we present the KKT optimality conditions of Problem (P1.2)
with the objective function specified in (\ref{eq:MMSE SINR}). For
any locally optimal power solution $\mv{p}^\ast$, there exist unique
Lagrangian multipliers $\mv \lambda=(\lambda_1,\cdots,\lambda_K)$
such that for any $k=1,\cdots,K$,\begin{align}& \mu_k\frac{\partial
R_k(\mv{p}^\ast)}{\partial p_k}+\sum\limits_{j\neq
k}\mu_j\frac{\partial R_j(\mv{p}^\ast)}{\partial p_k}=\lambda_k ,
\label{equ:KKT1}\\& \lambda_k(P_k^{{\rm max}}-p_k)=0,
\label{equ:KKT2}\\& \lambda_k\geq 0. \label{equ:KKT3}\end{align}

Now, for the price-based algorithm, define the price charged by
receiver $j$ to transmitter $k$, which indicates the sensitivity of
the achievable rate of receiver $j$ subject to the power change of
transmitter $k$, as
\begin{align}\pi_{j,k}=-\frac{\partial
R_j(\mv{p})}{\partial
p_k}=\frac{p_j\|\mv{h}_{j,j}^H(\sum\limits_{i\neq
j}p_i\mv{h}_{j,i}\mv{h}_{j,i}^H+\sigma_j^2\mv{I})^{-1}\mv{h}_{j,k}\|^2}{\ln2\left(1+p_j\mv{h}_{j,j}^H(\sum\limits_{i\neq
j}p_i\mv{h}_{j,i}\mv{h}_{j,i}^H+\sigma_j^2\mv{I})^{-1}\mv{h}_{j,j}\right)}.\label{eq:price}\end{align}Consequently,
we see that the KKT conditions in (\ref{equ:KKT1}), (\ref{equ:KKT2})
and (\ref{equ:KKT3}) are both necessary and sufficient for the
optimal solution to the following problem for user $k$,
$k=1,\cdots,K$:
\begin{align}\label{equ:SIMO price}\mathop{\mathtt{Maximize}}_{p_k} & ~~~
\mu_k\log_2\bigg(1+p_k\mv{h}_{k,k}^H(\sum\limits_{j\neq
k}p_j\mv{h}_{k,j}\mv{h}_{k,j}^H+\sigma_k^2\mv{I})^{-1}\mv{h}_{k,k}\bigg)
-p_k\sum\limits_{j\neq k}\mu_j\pi_{j,k} \nonumber \\
\mathtt{Subject \ to} & ~~~ p_k\leq P_k^{{\rm max}},\end{align}where
$p_j$ and $\pi_{j,k}$ are fixed, $\forall j\neq k$.

Similar to the ADP algorithm in \cite{Huang06}, we propose the
following algorithm to update the price and transmit power
iteratively for all users in SIMO-IC. Specifically, at each
iteration the algorithm does the following:\begin{itemize}\item[1.]
Each user announces its price obtained using (\ref{eq:price}) to all
the other users;
\item[2.] Each user updates its transmit power by solving Problem
(\ref{equ:SIMO price}), i.e.,
\begin{align}p_k=&\bigg[\frac{\mu_k}{\ln2\sum\limits_{j\neq
k}\mu_j\pi_{j,k}}-\frac{1}{\mv{h}_{k,k}^H(\sum\limits_{j\neq
k}p_j\mv{h}_{k,j}\mv{h}_{k,j}^H+\sigma_k^2\mv{I})^{-1}\mv{h}_{k,k}}\bigg]_0^{P_k^{{\rm
max}}}, ~~ \forall k, \end{align}\end{itemize}where
$[x]_a^b=\max(\min(x,b),a)$.

Because Problems (P1.2) and (\ref{equ:SIMO price}) possess the same
KKT optimality conditions, when the above algorithm converges to a
set of optimal solutions to problems in (\ref{equ:SIMO price}) for
all $k$'s, this set of solutions will be at least a locally optimal
solution to Problem (P1.2).

\subsubsection{Price-Based Algorithm for
MISO-IC}\label{appendix:interference pricing MISO-IC}

\ \ \

Next, we extend the ADP algorithm to MISO-IC. For any given transmit
beamforming vectors $\mv V$, we first define the price for user $k$
as
\begin{align}\label{equ:price}\pi_k&=-\frac{\partial R_k}{\partial
\Gamma_k}=\frac{\|\mv{h}_{k,k}^H\mv{v}_k\|^2}{\ln2(\|\mv{h}_{k,k}^H\mv{v}_k\|^2+\Gamma_k+\sigma_k^2)(\Gamma_k+\sigma_k^2)},\end{align}where
$\Gamma_k=\sum\limits_{j\neq k}\|\mv{h}_{k,j}^H\mv{v}_j\|^2$ is the
total interference power at the $k$th receiver. Let
$\mv{S}_k=\mv{v}_k\mv{v}_k^H$, $\forall k$. Given fixed interference
prices and beamforming vectors of all other users, the following
problem is to be solved by any user $k$ for its own transmit
beamforming update:\begin{align}\label{equ:game
theory}\mathop{\mathtt{Maximize}}_{\mv{S}_k} & ~~~
\mu_k\log_2(1+\frac{\mv{h}_{k,k}^H\mv{S}_k\mv{h}_{k,k}}{\Gamma_k+\sigma_k^2})-\sum\limits_{j\neq
k}\mu_j\pi_j\mv{h}_{k,j}^H\mv{S}_k\mv{h}_{j,k} \nonumber \\
\mathtt{Subject \ to} & ~~~{\rm Tr}(\mv{S}_k)\leq P_k^{{\rm max}}
\nonumber \\ & ~~~\mv{S}_k\succeq 0,\end{align}where
$\mv{S}_k\succeq 0$ means that $\mv{S}_k$ is a positive
semi-definite matrix. Similar to the previous case of SIMO-IC, we
can show that the KKT conditions of Problem (\ref{equ:game theory})
with $k=1,\cdots,K$ are also those of Problem (P1.3) by replacing
$\mv{v}_k\mv{v}_k^H$ with $\mv{S}_k$, $\forall k$. However, Problem
(P1.3) requires that the optimal solution $\mv{S}_k^\ast$ in Problem
(\ref{equ:game theory}) to be rank-one, which is not guaranteed a
priori. Thus, Problem (\ref{equ:game theory}) is a relaxation of the
original WSR maximization problem (P1.3) without considering the
rank-one constraint.

Interestingly, it was recently shown in \cite{Rui10} that the
optimal solution to Problem (\ref{equ:game theory}) is always of
rank-one, i.e., $\mv{S}_k^\ast=\mv{v}_k\mv{v}_k^H$. Hence, we
propose a price-based algorithm for MISO-IC in a similar way to that
for SIMO-IC. When this algorithm converges to a set of optimal
solutions to problems in (\ref{equ:game theory}) with
$k=1,\cdots,K$, this set of solutions are all rank-one and thus
corresponds to at least a locally optimal solution to Problem
(P1.3).

For this price-based algorithm for MISO-IC, the interference price
can be iteratively updated according to (\ref{equ:price}). As for
the update of beamforming vectors, we need to solve Problem
(\ref{equ:game theory}) for each user $k$. It can be verified that
Problem (\ref{equ:game theory}) is convex with strictly feasible
points, and thus it can be solved by the standard Lagrangian duality
method \cite{Boyd04} with a zero duality gap. The details of solving
Problem (\ref{equ:game theory}) can be found in Appendix
\uppercase\expandafter{\romannumeral1} of \cite{Rui10}, and are thus
omitted here.

\end{appendix}

\end{document}